\documentclass[envcountsect,envcountsame]{llncs}

\title{Testing Forest-Isomorphism in the Adjacency List Model}

\author{
  Mitsuru Kusumoto\inst{1}\thanks{JST, ERATO, Kawarabayashi Large Graph Project.}
  \and Yuichi
  Yoshida\inst{1, 2}\thanks{Supported by JSPS Grant-in-Aid for Research Activity Start-up (24800082), MEXT Grant-in-Aid for Scientific Research on Innovative Areas (24106003), and JST, ERATO, Kawarabayashi Large Graph Project.}
}

\institute{
  Preferred Infrastructure, Inc. \email{mkusumoto@preferred.jp}
  \and
  National Institute of Informatics. \email{yyoshida@nii.ac.jp}.
}

\usepackage{makeidx}  
\usepackage{algorithmicx}
\usepackage[ruled]{algorithm}
\usepackage[noend]{algpseudocode}
\usepackage{amsmath, amssymb, ascmac}
\usepackage[usenames,dvipsnames]{color}
\usepackage[dvips]{graphicx}
\usepackage{ulem}

\newcommand{\calD}{\mathcal{D}}

\newcommand{\calO}{\mathcal{O}}
\newcommand{\calU}{\mathcal{U}}

\newcommand{\calF}{\mathcal{F}}
\newcommand{\calT}{\mathcal{T}}

\newcommand{\bfF}{\mathbf{F}}

\newcommand{\bfw}{\mathbf{w}}
\newcommand{\bfu}{\mathbf{u}}
\newcommand{\bfv}{\mathbf{v}}
\newcommand{\bfone}{\mathbf{1}}

\newcommand{\poly}{\mathrm{poly}}
\newcommand{\polylog}{\mathrm{polylog}}
\newcommand{\rt}{\mathrm{root}}
\newcommand{\bbN}{\mathbb{N}}

\newcommand{\E}{\mathop{\mathbf{E}}}
\newcommand{\Var}{\mathop{\mathbf{Var}}}

\newcommand{\dist}{d}
\newcommand{\Freq}{\mathsf{\mathop{Freq}}}
\newcommand{\Size}{\mathsf{\mathop{Size}}}
\newcommand{\PReal}{\mathbb{R}_{\geq 0}}
\newcommand{\Real}{\mathbb{R}}
\newcommand{\ret}{\textbf{return} }
\newcommand{\vtx}{\mathbf{vtx}}

\newcommand{\Szemeredi}{Szemer{\'e}di}

\newcommand{\qI}{ {q_{\ref{lemma:approximate-freq}}}}
\newcommand{\qII}{ {q_{\ref{lemma:approximate-size}}}}

\begin{document}

\let\doendproof\endproof
\renewcommand\endproof{~\hfill\qed\doendproof}

\newcommand{\nqed}{~\hfill$\blacksquare$}

\maketitle

\normalem

\begin{abstract}
  We consider the problem of testing if two input forests are isomorphic or are far from being so.
  An algorithm is called an $\varepsilon$-tester for forest-isomorphism if 
  given an oracle access to two forests $G$ and $H$ in the adjacency list model,
  with high probability,
  accepts if $G$ and $H$ are isomorphic and rejects if we must modify at least $\varepsilon n$ edges to make $G$ isomorphic to $H$.
  We show an $\varepsilon$-tester for forest-isomorphism with a query complexity $\polylog(n)$ and a lower bound of $\Omega(\sqrt{\log{n}})$.
  Further, with the aid of the tester,
  we show that every graph property is testable in the adjacency list model with $\polylog(n)$ queries if the input graph is a forest.
\end{abstract}


\section{Introduction}
In \emph{property testing},
we want to design an efficient algorithm that distinguishes the case in which the input object satisfies some property or is ``far'' from satisfying it~\cite{Rubinfeld:1996um}.
In particular, an object is called \emph{$\varepsilon$-far} from a property $P$ if we have to modify an $\varepsilon$-fraction of the input to make it satisfy $P$.
A (randomized) algorithm is called an \emph{$\varepsilon$-tester} for a property $P$ if
it accepts objects satisfying $P$ and rejects objects that are $\varepsilon$-far from $P$ with high probability (say $2/3$).

Graph property testing is one of the major topics in property testing, and many properties are known to be testable in sublinear time or even in constant time (in the input size).
See~\cite{Goldreich:2010vk} for surveys.
In order to design sublinear-time testers, we have to define how to access the input graph, as just reading the entire graph requires linear time.
The model used here is the \emph{adjacency list model}~\cite{Kaufman:2004vg}.
In this model, the input graph $G=(V,E)$ is represented by an adjacency list and we are given an oracle access $\calO_G$ to it.
We have two types of queries.
The first query, called a \emph{degree query}, specifies a vertex $v$, and the oracle $\calO_G$ returns the degree of $v$.
The second query, called a \emph{neighbor query}, specifies a vertex $v$ and an index $i$, and the oracle $\calO_G$ returns the $i$-th neighbor of $v$.
A graph $G$ is called \emph{$\varepsilon$-far} from a property $P$ if we must add or remove at least $\varepsilon m$ edges for it to satisfy the property $P$,
where $m$ is the number of edges.
In contrast to other models such as the adjacency matrix and the bounded-degree models,
only a few properties are known to be efficiently testable in the adjacency list model.
For examples, 
testing triangle-freeness, $k$-colorability for a constant $k$, and bipartiteness requires $\Omega(\sqrt{n})$ queries~\cite{Alon:2008gn,BenEliezer:2008wz,Kaufman:2004vg},
where $n$ is the number of vertices.

A graph $G$ is called \emph{isomorphic} to another graph $H$ if there is a bijection $\pi:V(G) \to V(H)$ such that $(u,v)\in E(G)$ if and only if $(\pi(u),\pi(v)) \in E(H)$.
In this paper, we consider the problem of testing if the input graph $G$ is isomorphic to a fixed graph $H$, or if it is \emph{$H$-isomorphic}.
We assume that the (unknown) input graph $G$ has the same number of vertices as $H$.
The problem of deciding if a graph is isomorphic to $H$ is fundamental and theoretically important.
For example, the problem is one of the rare problems that is neither known to be in $\textbf{P}$ nor $\textbf{NP-Complete}$.
This motivates us to consider $H$-isomorphism in the property testing literature.
A \emph{graph property} refers to a property that is closed under taking isomorphism.
Then, $H$-isomorphism can be identified as the simplest graph property such that every graph property can be expressed as a union of $H$-isomorphisms.
Owing to these observations,
Newman and Sohler~\cite{Newman:2013hg} showed that every graph property is testable in the bounded-degree model if the input graph is a (bounded-degree) planar graph.
This connection also holds for the adjacency list model,
which motivates us to consider $H$-isomorphism in the adjacency list model.

If we assume that the input graph is an arbitrary graph possibly containing $\Omega(n^2)$ edges,
testing $H$-isomorphism in the adjacency list model requires $\Omega(\sqrt{n})$ queries~\cite{Fischer:2008tj}.
To investigate efficient testers for $H$-isomorphism, we restrict the input graph: We assume that the input graph and $H$ are forests with the same number of vertices $n$.
Note that we have no assumption on the degree as opposed to the bounded-degree model.
To avoid uninteresting technicalities,
we modify the definition of $\varepsilon$-farness as follows:
Instead of using the number of edges in $G$ to measure the distance,
we say that a forest $G$ is \emph{$\varepsilon$-far} from isomorphic to a forest $H$ if we must add or remove $\varepsilon n$ edges to transform $G$ to $H$.\footnote{Indeed, we often assume that the input graph contains $\Omega(n)$ edges in the adjacency list model. Thus, our definition of $\varepsilon$-farness for forests and the definition of $\varepsilon$-farness in the adjacency list model with the assumption are identical up to a constant multiplicative factor.}

With these definitions, we refer to the problem of testing the property of being isomorphic to a fixed forest as \emph{testing forest-isomorphism}.
The main result of this paper is as follows.
\begin{theorem}\label{thm:forest-isomorphism}
  In the adjacency list model, 
  we can test forest-isomorphism with $\polylog(n)$ queries.
\end{theorem}
Indeed, in our proof, we show that we can test forest-isomorphism even if both graphs are given as oracle accesses.

Further, we show a lower bound for testing forest-isomorphism.
\begin{theorem}\label{thm:lower-bound}
  In the adjacency list model,
  testing forest-isomorphism requires $\Omega(\sqrt{\log{n}})$ queries.
\end{theorem}

As a corollary of Theorem~\ref{thm:forest-isomorphism}, 
we show the following general result.
\begin{theorem}\label{thm:every-property-is-testable}
  In the adjacency list model,
  given an oracle access to a forest,
  we can test any graph property with $\polylog(n)$ queries.
\end{theorem}

\paragraph{Techniques}
We state a proof sketch of our main theorem, Theorem~\ref{thm:forest-isomorphism}.
Given a tree $G$,
by removing $\varepsilon n$ edges from $G$,
we can obtain a graph $G'$ with the following property for some $s = s(\varepsilon)$.
Each connected component of $G'$ is either (i) a tree of maximum degree at most $s$, or (ii) a tree consisting of a (unique) root vertex of degree more than $s$ and subtrees of size at most $s$.

The first step in our algorithm is providing an oracle access $\calO_{G'}$ to $G'$ using the oracle access $\calO_G$ to $G$.
We call $\calO_{G'}$ the \emph{partitioning oracle}.
In particular, if we specify a vertex $v$ and an index $i$, the oracle $\calO_{G'}$ returns whether the $i$-th edge incident to $v$ in $G$ is still alive in $G'$.
By carefully designing the construction of $G'$,
we can answer the query with $O(s^2)$ queries to $\calO_G$.


Suppose that we have an oracle access $\calO_{G'}$ to $G'$.
Since we can deal with trees of type (i) using existing algorithms in the bounded-degree model,
let us elaborate on trees of type (ii).
For a tree $T$ of type (ii), we can associate a tuple $(d,c_1,\ldots,c_{t(s)})$ with it,
where $t(s)$ is the number of possible trees of maximum degree at most $s$ and size at most $s$.
Note that $t(s)$ depends only on $\varepsilon$.
Here, $d$ is the degree of the root vertex of $T$,
and $c_i$ is the number of subtrees of the $i$-th type in $T$.
Though we cannot exactly compute the tuple,
given the root vertex of $T$,
we can approximate it well using $\calO_{G'}$.
Since $G'$ consists of trees of type (ii), we can associate a multiset of tuples with $G'$.
We call it the \emph{sketch} of $G'$.
Though we cannot exactly compute the sketch,
we can approximate it to some extent.
The query complexity becomes $\polylog(n)$ since we want to approximate $d$ to within the multiplicative factor of $1+\varepsilon$ and $d$ can be up to $n$.

If $G$ and $H$ are isomorphic, 
then sketches associated with $G'$ and $H'$ must be the same.
Our claim is that, if $G'$ and $H'$ are $\varepsilon$-far from being isomorphic,
then their sketches are also far.
Further, we will show that the distance between two sketches can be computed via maximum matching in the bipartite graph such that each vertex in the left part corresponds to a tree in $G'$ and each vertex in the right part corresponds to a tree in $H'$.
Since we can approximate sketches well and then approximate the size of the maximum matching from them, we obtain a tester for forest-isomorphism.

\paragraph{Related works}
There are two major models on the representation of graphs.
In the \emph{dense graph model}, a graph $G=(V,E)$ is given as an oracle $\calO_G:V \times V \to \{0,1\}$.
Given two vertices $u,v\in V$, the oracle returns whether $u$ and $v$ are connected in $G$.
A graph is called \emph{$\varepsilon$-far} from a property $P$ if we must add or remove at least $\varepsilon n^2$ edges for it to satisfy $P$.

In the dense graph model, many properties such as triangle-freeness and $k$-colorability are known to be testable in constant time~\cite{Goldreich:1998wa}.
Indeed, Alon et~al.~\cite{Alon:2009gn} obtained the characterization of constant-time testable properties using \Szemeredi's regularity lemma.
As for graph isomorphism, Fischer and Matsliah~\cite{Fischer:2008tj} showed that testing $H$-isomorphism can be carried out with $\widetilde{\Theta}(\sqrt{n})$ queries.
If both $G$ and $H$ are given as oracle accesses, then we need $\Omega(n)$ queries, and we can test with $\widetilde{O}(n^{5/4})$ queries.
We can trivially test forest-isomorphism:
If a graph is isomorphic to a forest $H$, then it has at most $n$ edges.
If a graph is $\varepsilon$-far from being isomorphic to $H$, then it has at least $\varepsilon n^2 - n$ edges (otherwise, we can remove all edges and then add new edges to make $H$).
Thus, we can distinguish the two cases only by estimating the number of edges up to, say $\frac{\varepsilon n^2}{2}$.

In the \emph{bounded-degree model} with a degree bound $d$,
a graph $G=(V,E)$ is given as an oracle $\calO_G:V \times [d] \to V \cup \{\bot\}$,
where $[d] = \{1,\ldots,d\}$ and $\bot$ is a special symbol.
Given a vertex $v \in V$ and an index $i \in [d]$, the oracle returns the $i$-th neighbor of $v$.
If there is no such neighbor, then the oracle returns $\bot$.

Many properties are known to be testable in constant time~\cite{Goldreich:2002bn} and several general conditions of constant-time testability are shown~\cite{Newman:2013hg,tanigawa2012testing}.
Hassidim et al.~\cite{Hassidim:2009ku} introduced the concept of the partitioning oracle to test minor closed properties.
Our partitioning oracle is similar to theirs, but their oracle provides an oracle access to the graph that is determined by its internal random coin whereas ours provides an oracle access to a graph that is deterministically determined.
As for graph isomorphism, it is known that $H$-isomorphism is testable in constant time when $H$ is hyperfinite~\cite{Newman:2013hg}.
Here, a graph is \emph{hyperfinite} if by removing $\varepsilon n$ edges, we can decompose the graph into connected components of size at most $f(\varepsilon)$ for some function $f$.

\paragraph{Organization}
In Section~\ref{sec:preliminaries}, we give notations and definitions used throughout the paper.
In Section~\ref{sec:partitioning-oracle}, we introduce the partitioning oracle.
Using the partitioning oracle, it suffices to consider the case where each tree in the input graph is either a bounded-degree tree or a tree consisting of a high-degree root and subtrees of small sizes.
In Section~\ref{sec:simple-case}, we consider the case in which every tree in the input graph is the latter type and the degrees of roots are within a small interval.
We deal with the general case and prove Theorem~\ref{thm:forest-isomorphism} in Section~\ref{sec:general-case}.
Due to limitations of space, some proofs in Section~\ref{sec:partitioning-oracle}, \ref{sec:simple-case}, \ref{sec:general-case} are presented in Appendix~\ref{sec:missing-po}, \ref{section:approx-alg-for-s-rooted-trees}, \ref{appendix:missing-general-case}.
We prove Theorem~\ref{thm:every-property-is-testable} in Appendix~\ref{appendix:every-property-is-testable}.
We show the lower bound in Appendix~\ref{sec:lower-bound}.

\section{Preliminaries}\label{sec:preliminaries}
For an integer $n$, we denote by $[n]$ the set $\{1, 2, \ldots, n\}$ and denote by $\bbN_{<n}$ (resp. $\bbN_{\le n}$) the set $\{0,1,\ldots,n-1\}$ (resp. $\{0,1,\ldots,n\}$).

Let $G=(V,E)$ be a graph.
For a vertex $v$, $\deg_G(v)$ denotes the \emph{degree} of $v$.
We omit the subscript if it is clear from the context.
For a set of vertices $S \subseteq V$, $G[S]$ denotes the subgraph \emph{induced} by $S$.
For graphs $G$ and $H$ with the same number of vertices, 
the \emph{distance} $d(G,H)$ between $G$ and $H$ is defined as the minimum number of edges that need to be added or removed to make $G$ isomorphic to $H$.
Formally, 
\begin{align*}
  d(G,H) = \min_{\pi} & (\#\{(u,v) \in E(G) \mid  (\pi(u),\pi(v)) \not \in E(H)\} \\
                             & + \#\{(u,v) \not \in E(G) \mid  (\pi(u),\pi(v))  \in E(H)\}),
\end{align*}
where $\pi$ is over bijections from $V(G)$ to $V(H)$.
We extend the definition of $d(G, H)$ for the case in which $G$ and $H$ have different number of vertices by adding a sufficient number of isolated vertices.
For a graph $G$ and an integer $k$, let $G+kv$ be the graph consisting of $G$ and $k$ isolated vertices.
If $|V(G)| > |V(H)|$, we define $d(G, H) = d(G, H+(|V(G)| - |V(H)|)v)$.
Similarly, if $|V(G)| < |V(H)|$, we define $d(G, H) = d(G+(|V(H)| - |V(G)|)v, H)$.

For an integer $s \ge 1$,
we call a tree $T$ an \emph{$s$-rooted tree} if $T$ contains a (unique) vertex $v$ with $\deg(v) \geq s+1$ such that each subtree of $v$ contains at most $s$ vertices.
The vertex $v$ is called the \emph{root vertex} of $T$ and is denoted by $\rt(T)$.
We call a tree $T$ an \textit{$s$-bounded-degree tree} if every vertex in $T$ has a degree of at most $s$.
We call a tree $T$ an \textit{$s$-tree} if it is an $s$-rooted tree or an $s$-bounded-degree tree.
To designate a union of trees, we use the term ``forest.''
For example, an \emph{$s$-rooted forest} means a disjoint union of $s$-rooted trees.

\section{Partitioning Oracle}\label{sec:partitioning-oracle}

In this section, we show that,
for any $\varepsilon > 0$,
there exists $s=s(\varepsilon)$ such that we can partition any forest into an $s$-forest by removing at most $\varepsilon n$ edges.
Then, we show that we can provide an oracle access to the $s$-forest, which we call the \emph{partitioning oracle}.
We refer to a vertex with degree more than $s$ in the original graph $G$ as a \emph{high-degree} vertex.

\newcommand{\sP}{ {s_{\ref{lmm:partition-to-s-forest}}}}
\newcommand{\qP}{ {q_{\ref{lmm:partition-to-s-forest}}}}
\newcommand{\tauP}{ {\tau_{\ref{lmm:partition-to-s-forest}}}}
\newcommand{\deltaP}{ {\delta_{\ref{lmm:partition-to-s-forest}}}}
\newcommand{\varepsilonP}{ {\varepsilon_{\ref{lmm:partition-to-s-forest}}}}

\newcommand{\Vhigh}{V_h}
\newcommand{\Vlow}{V_l}
\begin{lemma}[Partitioning oracle]\label{lmm:partition-to-s-forest}
  Suppose that we have an oracle access $\calO_G$ to a forest $G$ in the adjacency list model.
  Then for every $\varepsilon > 0$,
  we can provide an oracle access $\calO_G'$ to a graph $G'$ with the following properties:
  \begin{enumerate}
    \setlength{\itemsep}{0pt}
    \item $G'$ is an $s$-forest for some $s = \sP(\varepsilon)$. $G'$ depends only on $G$ and $\varepsilon$.
    \item $G'$ is obtained from $G$ by removing at most $\varepsilon n$ edges.
    \item Let $\Vhigh$ be high-degree vertices in $G$. Then, each tree in $G'$ contains at most one vertex from $\Vhigh$.
  \end{enumerate}
  The oracle $\calO'_G$ supports \emph{alive-edge queries}:
  Given a vertex $v$ and an integer $i$, the oracle returns whether the $i$-th edge incident to $v$ in $G$ still exists in $G'$.
  For each alive-edge query, the oracle issues $O(1/\varepsilon^2)$ queries to $\calO_G$.
  The output of $\calO'_G$ is deterministically calculated.
  Moreover,
  if $G$ and $H$ are isomorphic and $\Psi:V(G) \rightarrow V(H)$ is an isomorphism,
  $\calO'_G(e) = \calO'_H(\Psi(e))$ holds for every edge $e \in E(G)$.
\end{lemma}
\begin{proof}
  We set $s = \frac{11}{\varepsilon}$.
  If the degree of a vertex is at most $s$, we call it \emph{low-degree}.
  Let $\Vhigh$ and $\Vlow$ be the sets of high-degree and low-degree vertices in $G$, respectively.
  We call a connected component in $G[V_l]$ \emph{large} if it has more than $s$ vertices and \emph{small} otherwise.
  From the definition, there are at most $2n/s$ high-degree vertices in $G$ and at most $n/s$ large components in $G[V_l]$.

  We first give a polynomial-time algorithm that outputs an $s$-forest from the input forest $G$.
  First, we remove edges $(u,v)$ with $u,v \in \Vhigh$ from $G$.
  Owing to this, the resulting graph can be seen as a bipartite graph, where the left part is $\Vhigh$ and the right part consists of components in $G[V_l]$.
  Now for each small component $C$ in $G[V_l]$, if it is adjacent to two or more vertices in $\Vhigh$,
  we remove all the edges connecting $C$ and $\Vhigh$.
  Further, we remove all the edges between large components in $G[V_l]$ and $\Vhigh$.
  We define $G'$ as the resulting graph.
  As every subtree of each high-degree vertex is small, $G'$ is an $s$-forest.
  Since each connected component of $G'$ contains at most one high-degree vertex, the third property holds.
  Further, since any large small-degree connected component is not connected to a high-degree vertex, the first property holds.
  The total number of removed edges is at most $|V_h| + 2|V_h| + (n/s + 2|V_h|) = \varepsilon n$.
  Thus, the second property also holds.

  We next show how to provide an oracle access to $G'$.
  We can support alive-edge queries as follows:
  Let $e=(v,w)$ be the queried edge.
  For $v$ and $w$, we check if they are in $\Vhigh$, in a large component of $G[V_l]$, or in a small component of $G[V_l]$.
  If they are in a small component of $G[V_l]$,
  we check whether the component is incident to two or more vertices in $\Vhigh$.
  We can check these properties by performing a BFS in $G[V_l]$:
  If the BFS stops before visiting more than $s$ vertices, it means that the vertex belongs to a small component.
  Otherwise, the vertex belongs to a large component.
  From this information, we can answer the alive-edge query.
  The total number of queries to $\calO_G$ is $O(s^2)$.
  From the argument above, answers to alive-edge queries are determined deterministically.
\end{proof}

Since our construction of $G'$ is deterministic and we remove at most $\varepsilon n$ edges, the following corollary holds.

\begin{corollary}\label{corollary:through-oracle}
  Let $G$ and $H$ be two forests of $n$ vertices,
  and $G'$ and $H'$ be the graphs obtained from $G$ and $H$ by the partitioning oracle with a parameter $\frac{\varepsilon}{4}$, respectively.
  If $d(G,H)=0$, then $d(G',H')=0$ holds.
  If $d(G, H) \ge \varepsilon n$, then $d(G',H') \ge \varepsilon n/2$ holds.
  \qed
\end{corollary}

Thus,
we can preprocess the graph using the partitioning oracle,
and it is sufficient to show that we can test isomorphism between two $s$-forests.
We consider $s$-bounded-degree forests and $s$-rooted forests separately.
Therefore, we construct a tester for the isomorphism of each corresponding tree in $G'$ and $H'$.
To test isomorphism between $s$-bounded-degree forests, we use a technique from \cite{Newman:2013hg}.
We will develop a technique to test isomorphism between $s$-rooted forests in $G'$ and $H'$ under some conditions in the next section.

One technical issue of the partitioning oracle is
that we cannot obtain the exact degree $\deg_{G'}(v)$ of a vertex $v$ in $G'$
since $\deg_G(v)$ can be up to $n$.
Instead of computing the exact degree, we approximate the degree by randomly sampling incident edges as follows:
Choose $i \in [\deg_G(v)]$ uniformly at random
and apply the alive-edge query to the $i$-th incident edge.
For a parameter $q \ge 1$, repeat this $q$ times.
Then, count the number of existing edges. Let $c$ be this count.
We use the value $\frac{c\deg_G(v)}{q}$ as an approximation to $\deg_{G'}(v)$ and denote it by $\widetilde{\deg}_{G',q}(v)$.  
The standard argument using Chernoff's bound gives the following lemma.

\newcommand{\qD}{q_{\ref{lemma:approx-degree}}}
\begin{lemma}
  \label{lemma:approx-degree}
  Let $G'$ be the graph obtained from a graph $G$ by the partitioning oracle.
  For any $\delta,\tau \in (0,1)$ and a vertex $v$, there exists a polynomial $q = \qD(\delta, \tau)$ such that $\Pr[|\widetilde{\deg}_{G',q}(v) - \deg_{G'}(v)| \le \delta\deg_G(v)] \geq 1-\tau$.
\end{lemma}

There is another issue of the partitioning oracle.
If most parts of edges incident to a high-degree vertex $v$ (i.e., a vertex with degree more than $s$) are removed by the partitioning oracle,
the approximation $\widetilde{\deg}_{G',q}(v)$ may have a considerably large relative error.
However, we can ensure that the number of such high-degree vertices $v$ is sufficiently small.
To make the argument more formal, for an integer $R > s$, we call a vertex $v$ \emph{$R$-bad} if $R \cdot \max(\deg_{G'}(v), 1) \le \deg_{G}(v)$.
Otherwise, we call $v$ \emph{$R$-good}.
Note that an $R$-bad vertex must satisfy $\deg_{G}(v) \ge R > s$.
Thus, an $R$-bad vertex must be a high-degree vertex in $G$.
Further, we call an $s$-rooted tree \emph{$R$-bad} (resp. \emph{$R$-good}) if the root vertex is $R$-bad (resp. $R$-good).
Then, the number of vertices in $R$-bad $s$-rooted trees is bounded as follows.

\begin{lemma}
  \label{lemma:R-bad}
  Let $G'$ be the $s$-forest obtained from a graph $G$ by the partitioning oracle.
  For any $R > s$, the number of vertices in $R$-bad $s$-rooted trees of $G'$ is at most $\frac{4sn}{R}$.
\end{lemma}
\begin{proof}
  Let $B$ be the set of $R$-bad vertices and $B'$ be the set of vertices in $R$-bad $s$-rooted trees.
  Since there are at most $2n/R$ vertices with $\deg_G(v) \ge R$, $|B| \le 2n/R$ holds.
  From the third property of Lemma~\ref{lmm:partition-to-s-forest},
  each $s$-rooted tree in $G'$ contains at most one high-degree vertex in $G$.
  Hence,
  \[
  |B'| \le \sum_{v \in B} (s\cdot \deg_{G'}(v)+1)
  \le \sum_{v \in B} (\frac{s\deg_{G}(v)}{R}+1) \le \frac{4sn}{R}.
  \]
\end{proof}

By Lemma~\ref{lemma:R-bad}, random vertex sampling
does not pick up any $R$-bad vertex with high probability if $R$ is chosen sufficiently large.
In Section~\ref{sec:simple-case}, assuming that every s-rooted tree is $R$-good in the input graph,
we will construct a tester for forest-isomorphism.
In Section~\ref{sec:general-case}, combining Lemma~\ref{lemma:R-bad} and the tester given in Section~\ref{sec:simple-case},
we will construct a tester for any $s$-forest.

For later use, we define auxiliary procedures on $s$-rooted trees.
First, the following lemma is useful.
\begin{lemma}
  \label{lemma:find-root}
  Given a vertex $v \in V(G')$ in an $s$-rooted tree $T$,
  there is an algorithm that finds a root vertex $\rt(T)$ with query complexity $O(\poly(s))$.
\end{lemma}
\begin{proof}
  Perform a BFS in $G'$ starting from the vertex $v$ until we find a high-degree vertex.
  The third property of Lemma~\ref{lmm:partition-to-s-forest} guarantees that
  we can find the high-degree vertex and it is $\rt(T)$.
\end{proof}

Let $\mathcal{T}(s) = \{T^{(1)}, T^{(2)}, \ldots, T^{(t(s))}\}$ be the family of all rooted trees with at most $s$ vertices,
where $t(s) = |\mathcal{T}(s)|$.
For an $s$-rooted tree $T$, let $\Freq(T)$ be the $t(s)$-dimensional vector whose $i$-th coordinate is the number of subtrees of $\rt(T)$ isomorphic to $T^{(i)}$.
As the root vertex uniquely exists in an $s$-rooted tree $T$,
there is a unique $t(s)$-dimensional vector corresponding to $T$.

Since the degree of a root vertex can be up to $n$, we cannot exactly compute $\Freq(T)$.
Instead, we approximate $\Freq(T)$ by randomly sampling subtrees in $T$.
Given the root vertex $v$ of an $s$-rooted tree $T$,
we can define a procedure that approximates $\Freq(T)$.
We denote the procedure by $\widetilde\Freq_q(v)$.
The procedure $\widetilde\Freq$ randomly samples an edge incident to $v$ in $G$ (rather than $G'$) and invokes the alive-edge query.
If the edge is alive, the procedure performs a BFS from the edge to obtain the whole subtree rooted at the edge.
The procedure repeats this $q$ times, where $q$ is the parameter of the procedure.
We give the procedure $\widetilde\Freq$ in Algorithm~\ref{algorithm:estimate-freq}.
Again, Chernoff's bound guarantees the following.

\begin{algorithm}[!h]
  \caption{Given the root vertex $v$ of an $s$-rooted tree $T$ and an integer $q$,
  the procedure ${\widetilde{\Freq}}_q(v)$ returns an approximation to $\Freq(T)$ by randomly sampling subtrees in $T$.
  The integer $q$ represents the number of samples.}
  \label{algorithm:estimate-freq}
\begin{algorithmic}[1]
  \Procedure{${\widetilde{\Freq}}_q(v)$}{}
  \State{Let $\widetilde{\bfF}$ be the all-zero $t(s)$-dimensional vector.}
  \For{$j=1,\dots,q$}
    \State{Choose an integer $k$ from $[\deg_G(v)]$ uniformly at random.}
    \State{Ask whether the $k$-th edge $(v, u)$ incident to $v$ is alive.}
    \If{the edge is alive}
      \State{Perform a BFS from $u$ to obtain the whole subtree rooted at $u$.}
      \State{Suppose that the subtree is isomorphic to $T^{(i)}$. \label{alg-line:choose-t} Then, set $\widetilde{\bfF}[i] = \widetilde{\bfF}[i] + 1$.}
    \EndIf
  \EndFor

  \State{\ret $(\deg_G(v) / q) \cdot \widetilde{\bfF}$}
  \EndProcedure
\end{algorithmic}
\end{algorithm}

\begin{lemma}
  \label{lemma:approximate-freq}
  For $s\geq 1$ and $\delta,\tau \in (0,1)$,
  there exists a polynomial $q=\qI(s, \delta, \tau)$ such that
  for any $s$-rooted tree $T$,
  $|\Freq(T)[i] -\widetilde{\Freq}_q(\rt(T))[i]| \le \delta \deg_G(v)$
  for all $i \in [t(s)]$
  with probability at least $1-\tau$.
\end{lemma}
\begin{proof}
  Let $\qI(s, \delta, \tau) = O(\frac{\log(t(s)/\tau)}{\delta^2})$.
  By Chernoff's bound, it holds that $\Pr[|\Freq(T)[i] -\widetilde{\Freq}_\qI(v)[i]| > \delta \deg_G(v)] < \tau/t(s)$ for each $i$.
  By applying the union bound over all $i \in [t(s)]$, we obtain the lemma.
\end{proof}

It is also useful to approximate the number of vertices in an $s$-rooted tree.
For an $s$-rooted tree $T$, we can define a procedure $\widetilde{\Size}$
that approximates $|V(T)|$ by randomly sampling the subtrees of $T$ and computing the number of vertices in the subtrees.
We give the procedure $\widetilde{\Size}$ in Algorithm~\ref{alg:size}.
the procedure $\widetilde{\Size}$ first computes the approximate degree of the root vertex $v$ of $T$ by $\widetilde{\deg}$
with sufficiently large samples.
If $\widetilde{\deg} = 0$, the procedure just returns $1$ since $T$ looks an isolated vertex.
Otherwise, we randomly sample subtrees in $G'$ $q$ times, where $q$ is the parameter of the procedure.
For each subtree, we compute the number of vertices in the subtree.
To randomly sample the subtrees, we randomly choose an edge in $G$ (rather than $G'$) until we choose an alive edge.
This may take large amount of time since it is possible that most parts of edges incident to $v$ are not alive.
However, if $T$ is guaranteed to be $R$-good for some $R > s$, the following holds.

\begin{algorithm}[!h]
  \caption{Given two integers $q$, $R$ and the root vertex $v$ of an $R$-good $s$-rooted tree $T$,
    returns an approximation to $|V(T)|$ by randomly sampling the subtrees in $T$ and computing the size of the subtrees.
    The integer $q$ represents the number of samples.}
  \label{alg:size}
\begin{algorithmic}[1]
  \Procedure{$\widetilde{\Size}_{G', q, R}(v)$}{}
  \State{Set $q' = \qD(O(\delta/R), O(\tau))$ and compute $\tilde{d} = \widetilde{\deg}_{G', q'}(v)$.} \label{line:compute-tilde-d}
  \State{\textbf{if} $\deg_G(v) < R$ \textbf{then} round $\tilde{d}$ to the nearest integer.} \label{line:rounding}
  \State{\textbf{if} $\tilde{d} = 0$ \textbf{then} \ret 1} \label{line:isolated-vertex}
  \State{$\tilde{S} = 0$}
  \For{$j=1,\dots,q$}
    \Loop \label{line:loop-begin}
      \State{Choose an integer $k \in [\deg_G(v)]$ uniformly at random.}
      \State{Ask whether the $k$-th edge $(v,u)$ incident to $v$ is alive.}
      \State{\textbf{if }the edge is alive\textbf{ then break}}
    \EndLoop
    \State{Perform a BFS from $u$ to obtain the size $t$ of the subtree rooted at $u$.} \label{alg-line:calculate-size}
    \State{$\tilde{S} = \tilde{S} + t$} \label{line:loop-end}
  \EndFor
  \State{\textbf{return} $\tilde{d}\frac{\tilde{S}}{q} + 1$}
  \EndProcedure
\end{algorithmic}
\end{algorithm}

\begin{lemma}
  \label{lemma:approximate-size}
  For any $s, R\ge 1$ and $\delta,\tau \in (0,1)$,
  there exists a polynomial $q=\qII(s, \delta, \tau)$
  such that,
  for any $R$-good $s$-rooted tree $T$,
  $|\widetilde\Size_{G',q, R}(\rt(T)) - |V(T)| | \le \delta |V(T)|$ holds
  with probability at least $1-\tau$.
  The expected number of queries issued by the procedure $\widetilde\Size$
    is $O(\poly(s, R, \delta, \tau))$.
\end{lemma}

The proof of Lemma~\ref{lemma:approximate-size} is a little complicated.
We give the proof in Appendix~\ref{sec:missing-po}.

\newcommand{\Exact}{\mathrm{\mathop{Exact}}}
\newcommand{\Sketch}{\mathsf{\mathop{Sketch}}}
\newcommand{\Cell}{\mathrm{\mathop{Cell}}}
\newcommand{\Round}{\mathbf{\mathop{Round}}}
\newcommand{\size}{\mathrm{\mathop{size}}}
\newcommand{\qloop}{q_{\mathrm{\mathop{loop}}}}
\newcommand{\qfreq}{q_{\mathrm{\mathop{freq}}}}
\newcommand{\qsize}{q_{\mathrm{\mathop{size}}}}
\newcommand{\Capture}{\mathrm{\mathop{Capture}}}
\newcommand{\MM}{\mathcal{\mathop{M}}}
\newcommand{\MinMatch}{\mathrm{\mathop{MinMatch}}}

\newcommand{\bfz}{\mathbf{0}}
\newcommand{\const}{\mathrm{const}}
\newcommand{\mesh}{\mathbf{\mathop{mesh}}}
\newcommand{\ext}{\mathrm{\mathop{ext}}}

\newcommand{\qIV}{ {q_{\ref{lemma:approximate-count}}}}
\newcommand{\qIVloop}{ {q^{\ref{lemma:approximate-count}}_{\mathrm{loop}}}}
\newcommand{\qIVfreq}{ {q^{\ref{lemma:approximate-count}}_{\mathrm{freq}}}}
\newcommand{\qIVsize}{ {q^{\ref{lemma:approximate-count}}_{\mathrm{size}}}}

\newcommand{\qV}{ {q_{\ref{lemma:norm-of-count}}}}
\newcommand{\etaV}{ {\eta_{\ref{lemma:norm-of-count}}}}
\newcommand{\deltaV}{ {\delta'_{\ref{lemma:norm-of-count}}}}

\newcommand{\qrandom}{q_{\mathrm{\mathop{random}}}^{\ref{lemma:upper-bound-for-simple-case}}}

\section{When All Root Vertices Have Similar Degrees}\label{sec:simple-case}
In this section and the next section, we assume that we read the input graphs $G$ and $H$ through the partitioning oracle.
Thus, we are allowed to use alive-edge queries and the procedures $\widetilde\deg$, $\widetilde\Freq$, and $\widetilde\Size$.
Further, we assume that $s$ is a constant that depends only on $\varepsilon$.

We consider the case in which the root of all components have similar degrees.
Formally, we assume that each component in $G$ and $H$ is $R$-good $s$-rooted tree
and that the degree of each $s$-rooted tree in $G$ and $H$ is greater than $B$ and at most $\gamma B$.
Here, $B (> s)$ is an integer that can be up to $O(n)$ and $s, \gamma \ge 1$ is an arbitrary constant.
We call such a forest an \emph{$R$-good $s$-rooted forest with root degrees in $(B, \gamma B]$}.
In this section,
we will show that there is a forest-isomorphism tester for $R$-good $s$-rooted forest with root degrees in $(B, \gamma B]$
whose query complexity is a polynomial in $\gamma$ and $R$.

With the tester given in this section, we can construct a tester for the general case as follows.
After applying the partitioning oracle,
the graph becomes a disjoint union of an $s$-bounded-degree forest and an $s$-rooted forest.
We partition the $s$-rooted forest into several groups by the root degree.
First, we ignore all the $R$-bad $s$-rooted trees from the graph.
Since the number of $R$-bad trees is sufficiently small for a large $R$ from Lemma~\ref{lemma:R-bad}, this does not affect so much.
Second, if $\deg(\rt(T))$ is greater than $O(\gamma^i)$ and at most $O(\gamma^{i+1})$, we consider that a tree $T$ is in the $i$-th group.
Note that there are $O(\log{n})$ groups.
Then we apply the isomorphism tester of this section to each group.
If input graphs $G$ and $H$ are isomorphic, the tester must return YES (isomorphic) for all the groups.
In contrast, if $G$ and $H$ are $\varepsilon$-far from isomorphic, there must exist a group such that
the tester returns NO (not isomorphic) for the group.
Here, there is one technical issue: The number of vertices in such a group might be different.

We resolve this issue.
We assume that $n := |V(G)|$ and $n':=|V(H)|$ might be slightly different
and the algorithm does not know the exact values of $n$ and $n'$ but know their approximations.
Formally, we assume that our algorithm will be given a value $\tilde{n} \ge 1$, an approximation to $n$ and $n'$,
and $\eta \in (0,1)$ with $\frac{\tilde{n}}{n}, \frac{\tilde{n}}{n'} \in [1-\eta, 1]$.

We can prove the following lemma.

\newcommand{\etaI}{\eta_{\ref{lemma:upper-bound-for-simple-case}}}
\begin{lemma}
  \label{lemma:upper-bound-for-simple-case}
  Suppose that we are given $\varepsilon'>0$, $\tilde{n} \ge 1$, $\gamma \ge 1$, $R, B > s$, $\tau \in (0,1)$
  and we can access $s$-forests $G$ and $H$ through the partitioning oracle,
  where $n = |V(G)|$ and $n' = |V(H)|$ might be different.
  Then, there exists $\eta = \etaI(s, \varepsilon', \gamma, \tau, R)>0$ with the following property.
  If $G$ and $H$ are $R$-good $s$-rooted forests with root degrees in $(B, \gamma B]$ with $\frac{\tilde{n}}{n}, \frac{\tilde{n}}{n'} \in [1-\eta, 1]$,
  then there exists an algorithm that tests if $d(G, H) = 0$ or $d(G, H) \ge \varepsilon' \tilde{n}$
  with probability at least $1-\tau$.
  Assuming that $s$ is constant, the query complexity is a polynomial in $R, \gamma, \varepsilon', \tau$ and does not depend on $B, \tilde{n}$.
  Further, denote by $\qrandom(s, \gamma, \varepsilon', \tau)$ the number of random vertex queries the algorithm invokes.
  Then, $\qrandom$ is a polynomial in $\gamma, \varepsilon', \tau$.
\end{lemma}

In this section, we only write an overview of the proof of Lemma~\ref{lemma:upper-bound-for-simple-case}
since the proof is complicated,
We provide the proof in Appendix~\ref{section:approx-alg-for-s-rooted-trees}.

Since $\Freq(T)$ maps to a unique $t(s)$-dimensional vector corresponding to an $s$-rooted tree $T$,
there is a unique multiset of vectors corresponding to an $s$-rooted forest $G$.
For a $t(s)$-dimensional vector $\bfw \in \bbN_{<n}^{t(s)}$, let $\Psi_{G}[\bfw]$ be the number of $s$-rooted trees $T$ in $G$ such that $\Freq(T) = \bfw$.
Note that $\Psi_G$ can be seen as the sketch of $G$.
Clearly, $G$ is isomorphic to $H$ if and only if $\Psi_{G}[\bfw] = \Psi_{H}[\bfw]$ for all $\bfw$.
We use this property to create a tester.
Since it is impossible to compute $\Psi_G$ exactly,
we resort to approximate it.
We choose an integer $k \geq 1$, and
divide each axis of the $t(s)$-dimensional space into $k$ segments to make $k^{t(s)}$ cells.
We then estimate the number of $s$-rooted trees in each cell.
We call this estimation the \emph{sketch} of $G$.
We focus on computing the sketch.

For an integer $k \ge 1$, we define intervals $I_i = [\frac{\tilde{n}i}{(1-\eta)k}, \frac{\tilde{n}(i+1)}{(1-\eta)k})$ ($i \in \bbN_{<k}$).
Note that, for every $0\le i \le n-1$, there exists a unique interval $I_j$ with $i \in I_j$.
For a vector $\bfu \in \bbN_{<k}^{t(s)}$,
let $\Cell(\bfu)$ be the corresponding cell formed by intervals $I_{\bfu[1]},\ldots,I_{\bfu[t(s)]}$.
Further, for a vector $\bfw \in [0,n]^{t(s)}$,
we define $\Round(\bfw) = \bfu$, where $\bfu \in \bbN_{<k}^{t(s)}$ is such that $\Cell(\bfu) \ni \bfw$.

For a vector $\bfu \in \bbN_{<k}^{t(s)}$,
we approximate the number of $s$-rooted trees $T$ in $G$ with $\Freq(T) \in \Cell(\bfu)$
by the following algorithm $\widetilde{\Sketch}$.

\begin{algorithm}[!h]
  \caption{returns a map $\Phi : \bbN_{<k}^{t(s)} \rightarrow [0, n]$, given integers $\qloop$, $\qfreq$, $\qsize$,$R$,$k$, a real $\tilde{n}$ and an $R$-good $s$-rooted forest $G$ with root degrees in $(B, \gamma B]$ through the partitioning oracle.
  Here, $\Phi(\bfu)$ is an approximation to the number of $s$-rooted trees $T$ with $\Freq(T) \in \Cell(\bfu)$.}
  \label{alg:sketch}
\begin{algorithmic}[1]
  \Procedure{$\widetilde{\Sketch}_{\qloop, \qfreq, \qsize ,R,k}(G)$}{}
  \State{Set $\Phi(\bfu) = 0$ for all $\bfu \in \bbN_{<k}^{t(s)}$}
  \For{$j=1,\dots,\qloop$}
    \State{Choose a vertex $u \in V(G)$ uniformly at random} \label{alg-line:chose-u}
    \State{Perform a BFS from $u$ to find a root vertex $v$.}
    \State{$\bfu = \Round(\widetilde{\Freq}_{\qfreq}(v))$} \label{alg-line:rounding}
    \State{$\Phi(\bfu) = \Phi(\bfu) + 1/\widetilde{\Size}_{G, \qsize, R}(v)$}  \label{alg-line:approx-inverse-size}
  \EndFor
  \State{\textbf{return} $\frac{\tilde{n}}{\qloop} \Phi$} \label{alg-line:return-phi}
  \EndProcedure
\end{algorithmic}
\end{algorithm}

\newcommand{\qIII}{ {q_{\ref{lemma:inverse-size}}}}
To create a forest-isomorphism tester,
we first compute the sketches of $G$ and $H$ by the algorithm~$\widetilde{\Sketch}$,
and then, we compute the minimum matching between the sketches.
Here, the minimum matching is defined as the min-cost flow of complete bipartite graphs
where vertices correspond to the cells of the sketches and the weight of an edge is the L1 distance between two cells of the sketches in the $t(s)$-dimensional space.
Since the L1 distance in the $t(s)$-dimensional space corresponds to the number of different subtrees in $s$-rooted trees,
we can prove that the a minimum matching between the sketches
is a good approximation to $d(G, H)$ with high probability.
Thus, it suffices to compute the sketches of $G$ and $H$ and the minimum matching between them.
Note that we do not have to make any query to $G$ and $H$ to compute the minimum matching.
 

\newcommand{\qX}{ {q_{\ref{lemma:mmgh-mmtcgtch}}}}

\newcommand{\qVI}{ {q_{\ref{lemma:count-tildecount-diff}}}}
\newcommand{\deltaVI}{ {\delta_{\ref{lemma:count-tildecount-diff}}}}

\newcommand{\qVII}{ {q_{\ref{lemma:g-count-diff}}}}
\newcommand{\qVIIfreq}{ {q_{\mathrm{freq}\ref{lemma:g-count-diff}}}}
\newcommand{\deltaVII}{ {\delta_{\ref{lemma:g-count-diff}}}}
\newcommand{\kI}{ {k_{\ref{lemma:g-count-diff}}}}
 
\newcommand{\TestSimple}{\mathsf{\mathop{TestRootedForest}}}

\newcommand{\GP}[1]{{G}^{[#1]}}
\newcommand{\HP}[1]{{H}^{[#1]}}

\newcommand{\WhichComponent}{\mathsf{\mathop{Which}}}
\newcommand{\Which}{\mathsf{\mathop{Which}}}
\newcommand{\Random}{\mathsf{\mathop{Random}}}
\newcommand{\TestIsomorphism}{\mathsf{\mathop{TestIsomorphism}}}
\newcommand{\TestBounded}{\mathsf{\mathop{TestBoundedDegreeForest}}}

\newcommand{\TestEach}{\mathsf{\mathop{TestForestOfSameType}}}

\section{General Case}
\label{sec:general-case}
In this section, we prove Theorem~\ref{thm:forest-isomorphism}.
Missing parts of this section are given in Appendix~\ref{appendix:missing-general-case}.
Missing proofs are given in Appendix~\ref{ss:missing-general-case}.
Again $G$ and $H$ denote the graphs given through the partitioning oracle
and $s$ is constant.
For an integer $L \ge 1$, we call $G_1, \dots, G_L \subseteq G$ \emph{a partition of $G$}
if each $G_i$ is a union of connected components in $G$ and $G$ is a disjoint union of $G_1, \dots, G_L$.
The following lemma allows us to consider each part in the partition separately.

\begin{lemma}
  \label{lemma:farness-test}
  Let $L \ge 1$ be an integer and $G_1,\cdots,G_L$ (resp. $H_1,\cdots,H_L$) be any partition of $G$ (resp. $H$).
  Then, for any $\beta_1, \cdots, \beta_L \ge 0$ summing up to $1$, the following holds:
  For any $\varepsilon > 0$,
  if $d(G, H) \ge \varepsilon n$,
  there exists $i \in [L]$ such that 
  $d(G_i, H_i) \ge \beta_i \varepsilon n$ holds.
\end{lemma}
\begin{proof}
  We can obtain the lemma immediately from the following claim.

  \begin{claim}
    \label{claim:d-ubound-for-general-case}
    $d(G, H) \le \sum_{i=1}^{L} d(G_i, H_i)$.
  \end{claim}

  We prove the claim.
  Construct a sequence of modifications to transform $G$ to $H$.
  For each subgraph $G_i$ with $|V(G_i)| \ge |V(H_i)|$,
  we transform $G_i$ into $H_i$ and $|V(G_i)| - |V(H_i)|$ isolated vertices.
  After this modification, for each subgraph $G_i$ with $|V(G_i)| < |V(H_i)|$,
  we use $G_i$ and $|V(H_i)| - |V(G_i)|$ isolated vertices to construct $H_i$.
  The total number of modifications is $\sum_{i} d(G_i, H_i)$.
\end{proof}

To construct a tester for the isomorphism of $s$-forests,
we first give a partition of an $s$-forest and apply Lemma~\ref{lemma:farness-test}.
Then we test the isomorphism of each corresponding partition of $G$ and $H$.
That is, we check $d(G_i, H_i) = 0$ or $d(G_i, H_i) \ge \beta_i \varepsilon n$ for each $i$.
Here, if $d(G, H)=0$, all parts of the partition in $G$ and $H$ are isomorphic, so all the tests must output YES (with high probability).
If $d(G, H) \ge \varepsilon n$, there must be an index $i$ where the test outputs NO.
To provide oracle accesses to $G_i$ and $H_i$, we estimate the size of $V(G_i)$ and $V(H_i)$ by random sampling.
If they are sufficiently far, we immediately return NO.
If they are sufficiently small, we simply ignore $G_i$ and $H_i$.
Otherwise, we can provide the oracle accesses to $G_i$ and $H_i$ that costs for each query at most $\poly(L)$ queries to $G$ and $H$.
Using this access, we test whether $d(G_i, H_i) \ge \beta_i \varepsilon n$.

To provide a partition of an $s$-forest, we introduce a new notion.
For $\alpha,\gamma \ge 1$, $\mu > 0$, and a tree $T$, we say that $T$ is \textit{on the $(\alpha, \gamma, \mu)$-boundary},
if there exists an integer $i \ge 1$ with $1-\mu \le \deg(\rt(T)) / (\alpha \gamma^i) \le 1+\mu$.
We denote by $B_{\alpha, \gamma, \mu}(G)$
the number of vertices in the trees of $G$ that are on the $(\alpha, \gamma, \mu)$-boundary.
For $\lambda > 0$, we call $\alpha$ \emph{$(\gamma,\mu, \lambda)$-good with respect to $G$} if $B_{\alpha,\gamma,\mu}(G) < \lambda n$.
We can show that, if we choose $\alpha$ from $[1, \gamma]$ at random,
$\alpha$ is $(\gamma,\mu, \lambda)$-good with high probability.

\begin{lemma}
  \label{lemma:boundary-bound}
  Suppose that $\alpha$ is chosen from $[1, \gamma]$ uniformly at random.
  Then, for $\gamma \ge 2$, $\mu \in (0,1/3)$, and $\lambda \in (0,1)$, $\alpha$ is 
  $(\gamma,\mu, \lambda)$-good with respect to $G$ with probability at least $1-\frac{4\gamma \mu}{\lambda}$.
\end{lemma}

We consider a partition of an $s$-forest $G$.
Let $\alpha$, $\gamma$, $\mu$, and $R$ be values chosen later.
Let $\GP{0}_{s, \alpha, \gamma, \mu, R}$ be the maximal $s$-bounded-degree forest in $G$
and $\GP{1}_{s, \alpha, \gamma, \mu, R}$ be the union of $R$-good $s$-rooted trees with root degree in $(s, \alpha \gamma]$ that are not on the $(\alpha, \gamma, \mu)$-boundary in $G$.
Similarly, for $2 \le i \le L$, where $L=\lceil \log{n} / \log{\gamma} \rceil$,
let $\GP{i}_{s, \alpha, \gamma, \mu, R}$ be the union of $R$-good $s$-rooted trees with root degree in $(\alpha \gamma^{i-1}, \alpha \gamma^{i}]$ that are not on the $(\alpha, \gamma, \mu)$-boundary in $G$.
Finally, let $\GP{L+1}_{s, \alpha, \gamma, \mu, R}$ be the remaining trees that are not assigned to any partition so far.
That is, $\GP{L+1}$ is the union of trees that are $R$-bad or on the $(\alpha, \gamma, \mu)$-boundary in $G$.
We omit the subscript of $\GP{i}_{s, \alpha, \gamma, \mu, R}$ if it is clear from the context.
Note that we can write $G = \GP{0} \cup \GP{1} \cup \cdots \cup \GP{L+1}$.
We use the same notion for the other graph $H$.

We define a procedure that, given a vertex $v \in V(G)$, returns $i$ with $v \in \GP{i}$ as follows.
Our procedure first determines if $v$ is in an $s$-bounded-degree tree by performing a BFS from $v$ until we visit $O(s)$ vertices.
If we cannot find a high-degree vertex, $v \in \GP{0}$.
Otherwise, for a parameter $q \ge 1$, we invoke $\widetilde{\deg}_q(\rt(v))$ and return an appropriate output.
We call this procedure $\Which_q(v)$.

Here, the technical issue is that the procedure $\Which$ may output a wrong value.
We show that $\Which$ outputs the correct value with high probability for any partition of $G$ except for $\GP{L+1}$
and that the size of $\GP{L+1}$ is sufficiently small.

\newcommand{\qW}{   {q_{\ref{lemma:which-component}}}}
\begin{lemma}
  \label{lemma:which-component}
  For any $\tau \in (0,1)$ and $R \ge 1$,
  there exists a polynomial $q = \qW(\gamma, \mu, R, \tau)$
  such that the procedure $\Which_q(v)$ outputs a correct value with probability $1 - \tau$ for $v \in V(\GP{0}_{s, \alpha, \gamma, \mu, R}) \cup \cdots \cup V(\GP{L}_{s, \alpha, \gamma, \mu, R})$.
  
\end{lemma}

\begin{lemma}
  \label{lemma:size-of-garbage}
  For any $\gamma \ge 2$ and $\lambda \in (0,1)$,
  there exist $R = O(s/\lambda)$, $\mu = O(\lambda / \gamma)$ such that
  if $\alpha$ is chosen from $[1, \gamma]$ uniformly at random,
  $|V(\GP{L+1}_{s, \alpha, \gamma, \mu, R})| \le \lambda n$ holds with probability $1 - O(1)$.
\end{lemma}

\newcommand{\qwhich}{q_{\mathrm{\mathop{which}}}}
Using the procedure $\WhichComponent$, we can approximate the number of vertices in $\GP{i}$ by random sampling.
For $i \in \bbN_{\le L}$, we denote by $\Size_{\qloop, \qwhich}(G, i)$ the algorithm
that samples $\qloop$ vertices uniformly at random, and applies $\Which_{\qwhich}$
for each sampled vertex, and then approximates $|V(\GP{i})|$.
By Chernoff's bound, we obtain the following lemma.

\newcommand{\qloopC}{  {q_{\mathrm{\mathop{loop}}\ref{lemma:component-size}}}}
\newcommand{\qwhichC}{ {q_{\mathrm{\mathop{which}}\ref{lemma:component-size}}}}
\begin{lemma}
  \label{lemma:component-size}
  For any $\delta,\tau \in (0,1)$ and parameters $\alpha$, $\gamma$, $\mu$, and $R$, 
  there exist polynomials
  $\qloop=\qloopC(\delta, \tau)$ and $\qwhich = \qwhichC(\delta, \tau)$
  such that the following holds:
  For any $\lambda \in (0,1)$ with $|V(\GP{L+1}_{s, \alpha, \gamma, \mu, R})| \le \lambda n$,
  $|\Size_{\qloop, \qwhich}(G, i) - |V(\GP{i}_{s, \alpha, \gamma, \mu, R})| | \le (\lambda + \delta) n$ with probability $1-\tau$.
  \qed
\end{lemma}

Further, using the procedure $\Which$, we can provide oracle accesses to $\GP{i}$ for $i \in \bbN_{i \le L}$.
Let $\Random_q(G, i)$ denote the procedure
that repeats itself to pick up a vertex $v$ in $G$ uniformly at random and invokes the procedure $\Which_q(v)$
and returns $v$ if the returned value of $\Which$ is $i$.

\newcommand{\lambdaR}{ {\lambda_{\ref{lemma:provide-random-access}}}}
\newcommand{\qR}{ {q_{\ref{lemma:provide-random-access}}}}
\begin{lemma}
  \label{lemma:provide-random-access}
  For every $\delta,\tau \in (0,1)$ and parameters $\alpha$, $\gamma$, $\mu$, and $R$, there exist polynomials $q = \qR(\delta, \tau)$ and $\lambda = \lambdaR(\delta, \tau)$ such that the following holds for every $i \in \bbN_{\le L}$:
  If $|V(\GP{i})| \ge \delta n$ and $|V(\GP{L+1})| \le \lambda n$,
  the procedure $\Random_q(G, i)$ outputs a vertex of $\GP{i}$ uniformly at random
  by invoking the procedure $\Which_q$ at most $O(1/(\delta \tau))$ times
  with probability $1-\tau$.
\end{lemma}

The sketch of the proof of Theorem~\ref{thm:forest-isomorphism} is as follows.
As we mentioned, 
it suffices to create an isomorphism tester between $\GP{i}$ and $\HP{i}$
for each $i \in \bbN_{\le L}$.
First, set $\gamma = 2s$ and choose $\alpha \in [1, \gamma]$ uniformly at random.
From Lemma~\ref{lemma:size-of-garbage}, $|V(\GP{L+1})|$ and $|V(\HP{L+1})|$ are small with high probability.
Thus, we can apply the procedures $\Which$, $\Size$ and $\Random$ to the input graphs.
From Lemmas~\ref{lemma:which-component}, \ref{lemma:component-size}, and \ref{lemma:provide-random-access}, these procedures output the correct value with sufficiently high probability.
Using the procedure $\Size$,
we can test if $|V(\GP{i})|$ and $|V(\HP{i})|$ are large and sufficiently close.
Then, we can test
forest-isomorphism between $\GP{i}$ and $\HP{i}$ (with high probability)
by providing oracle accesses to $\GP{i}$ and $\HP{i}$ through the procedure $\Random$.
For $i=0$, we use a method proposed by \cite{Newman:2013hg}
with a little modification. See Appendix~\ref{ss:s-bounded-degree} for details.
For $1 \le i \le L$, we use the method in Section~\ref{sec:simple-case}.
Here, every parameter depends on $\polylog(n)$ assuming that $s$ is constant.
Thus, the query complexity of our forest-isomorphism tester is $\polylog(n)$ in total.
See Algorithm~\ref{algorithm:test-isomorhism} in Appendix~\ref{ss:proof-of-isomorphism-tester} for the detailed description of the tester for forest-isomorphism.

\newcommand{\etaII}{\eta_{\ref{lemma:s-bounded-degree}}}
\newcommand{\DII}{D_{\ref{lemma:s-bounded-degree}}}
\newcommand{\deltaII}{\delta''_{\ref{lemma:s-bounded-degree}}}
\newcommand{\Dist}{\mathsf{\mathop{Dist}}}

\newcommand{\TS}[1]{\tilde{z}_{#1}}

\bibliographystyle{abbrv}
\bibliography{isomorphism}

\begin{thebibliography}{10}

\bibitem{Alon:2009gn}
N.~Alon, E.~Fischer, I.~Newman, and A.~Shapira.
\newblock A combinatorial characterization of the testable graph properties:
  It's all about regularity.
\newblock {\em SIAM Journal on Computing}, 39(1):143--167, 2009.

\bibitem{Alon:2008gn}
N.~Alon, T.~Kaufman, M.~Krivelevich, and D.~Ron.
\newblock Testing triangle-freeness in general graphs.
\newblock {\em SIAM Journal on Discrete Mathematics}, 22(2):786--819, 2008.

\bibitem{BenEliezer:2008wz}
I.~Ben-Eliezer, T.~Kaufman, M.~Krivelevich, and D.~Ron.
\newblock Comparing the strength of query types in property testing: the case
  of testing $k$-colorability.
\newblock In {\em SODA'08: Proceedings of the 19th Annual ACM-SIAM Symposium on
  Discrete Algorithms}, pages 1213--1222, 2008.

\bibitem{Fischer:2008tj}
E.~Fischer and A.~Matsliah.
\newblock Testing graph isomorphism.
\newblock {\em SIAM Journal on Computing}, 38(1):207--225, 2008.

\bibitem{Goldreich:2010vk}
O.~Goldreich.
\newblock Introduction to testing graph properties.
\newblock pages 105--141. Property Testing, 2010.

\bibitem{Goldreich:1998wa}
O.~Goldreich, S.~Goldwasser, and D.~Ron.
\newblock Property testing and its connection to learning and approximation.
\newblock {\em Journal of the ACM}, 45(4):653--750, 1998.

\bibitem{Goldreich:2002bn}
O.~Goldreich and D.~Ron.
\newblock Property testing in bounded degree graphs.
\newblock {\em Algorithmica}, 32(2):302--343, 2002.

\bibitem{Hassidim:2009ku}
A.~Hassidim, J.~A. Kelner, H.~N. Nguyen, and K.~Onak.
\newblock Local graph partitions for approximation and testing.
\newblock {\em FOCS'09: Proceedings of the 50th Annual IEEE Symposium on
  Foundations of Computer Science}, pages 22--31, 2009.

\bibitem{Kaufman:2004vg}
T.~Kaufman, M.~Krivelevich, and D.~Ron.
\newblock Tight bounds for testing bipartiteness in general graphs.
\newblock {\em SIAM Journal on Computing}, 33(6):1441--1483, 2004.

\bibitem{Newman:2013hg}
I.~Newman and C.~Sohler.
\newblock Every property of hyperfinite graphs is testable.
\newblock {\em SIAM Journal on Computing}, 42(3):1095--1112, 2013.

\bibitem{Rubinfeld:1996um}
R.~Rubinfeld and M.~Sudan.
\newblock Robust characterizations of polynomials with applications to program
  testing.
\newblock {\em SIAM Journal on Computing}, 25(2):252--271, 1996.

\bibitem{tanigawa2012testing}
S.~Tanigawa and Y.~Yoshida.
\newblock Testing the supermodular-cut condition.
\newblock {\em Algorithmica}, pages 1--11, 2013.

\bibitem{Wu:2013yw}
Y.~Wu, Y.~Yoshida, Y.~Zhou, and A.~Vijayraghavan.
\newblock Graph isomorphism: Approximate and robust, 2013.
\newblock manuscript.

\end{thebibliography}

\appendix
\section{Proof of Lemma~\ref{lemma:approximate-size}}
\label{sec:missing-po}

\begin{proof}
  First, we evaluate the probability that our procedure returns a good approximation.
  Let $S = |V(T)|-1$ and set $q = \qII = O(\frac{s^2 \log(1/\tau)}{\delta^2})$.
  In Line~\ref{line:loop-begin}--\ref{line:loop-end}, we choose an edge incident to $v$ in $G'$ uniformly at random.
  Therefore, by Chernoff's bound, $|\frac{\tilde{S}}{q} - \frac{S}{\deg_{G'}(v)}| \le O(\delta)$ holds
  with probability $1-O(\tau)$.
  Since we assume that $T$ is $R$-good, at least one of (i) $R\deg_{G'}(v) > \deg_G(v)$ and (ii) $\deg_G(v) < R$ holds.
  To bound $| \widetilde{\deg}_{G', q'}(v) \frac{\tilde{S}}{q} - S |$,
  let us consider these two cases.

  When (i) holds but (ii) does not hold,
  $v$ is not an isolated vertex in $G'$ and $\tilde{d}$ in the procedure is equal to the output of $\widetilde{\deg}_{G', q'}(v)$.
  From Lemma~\ref{lemma:approx-degree},
  $| \widetilde{\deg}_{G', q'}(v) - \deg_{G'}(v) | \le O(\deg_{G}(v)/R)$ holds with probability $1-O(\tau)$.
  The following claim is useful.

\begin{claim}
  \label{claim:four-reals}
  For any positive reals $A,B,C,D$ with $|A-B| \le \alpha$ and $|C-D| \le \beta$,
  $|AC-BD| \le \alpha D + \beta B + \alpha \beta$ holds.
\end{claim}
\begin{proof}
  By the triangle inequality,
  \begin{eqnarray*}
  |AC-BD| &=& \frac{1}{2}|(A-B)(C+D) + (A+B)(C-D)| \\
  &\le& \frac{1}{2}\left( |A-B| |C+D| + |A+B| |C-D| \right) \\
  &\le& \frac{1}{2}\left( \alpha |C+D| + \beta |A+B| \right)
  \le \alpha D + \beta B + \alpha \beta.
  \end{eqnarray*}
\end{proof}

  From the claim and the union bound,
  we have $| \widetilde{\deg}_{G', q'}(v) \frac{\tilde{S}}{q} - S |
  \le O(\delta \deg_{G'}(v)) + O(\frac{\delta \deg_G(v)}{R \deg_{G'}(v)} S) + O(\delta \deg_G(v)/R)
  \le O(\delta S) + O(\delta S) + O(\delta S) = \delta S$
  with probability $1-\tau$.

  When (ii) holds,
  $\tilde{d} = \deg_{G'}(v)$ holds (with probability $1-O(\tau)$) by rounding in Line~\ref{line:rounding}.
  Thus, if $v$ is an isolated vertex in $G'$, our procedure will return $1$ in Line~\ref{line:isolated-vertex}.
  Otherwise, $| \tilde{d}\frac{\tilde{S}}{q} - S | =
  | \frac{\tilde{S}}{q} - \frac{S}{\deg_{G'}(v)} | \cdot \tilde{d}
  \le \delta \tilde{d} \le \delta |V(T)|$ holds with probability $1-\tau$.

  Next, we turn to analyze the expected number of queries issued by the procedure.
  Since $q'$ is $\poly(R, \delta, \tau)$, we make at most $\poly(R, \delta, \tau)$ queries
  to compute $\widetilde{\deg}_{G', q'}(v)$ in Line~\ref{line:compute-tilde-d}.
  Further, since we assume that $T$ is $R$-good,
  Line~\ref{line:loop-begin}--\ref{line:loop-end} takes $O(R+\poly(s))$ time on average.
  Thus, the expected query complexity is $O(\poly(s, R, \delta, \tau))$ in total.
\end{proof}

\section{Proof of Lemma~\ref{lemma:upper-bound-for-simple-case}}
\label{section:approx-alg-for-s-rooted-trees}

In this section, we prove Lemma~\ref{lemma:upper-bound-for-simple-case}.
We use the notions defined in Section~\ref{sec:simple-case}.
Throughout this section, $c(G)$ denotes the number of connected components in $G$.

This section is organized as follows.
First, we analyze the behavior of the algorithm~$\widetilde{\Sketch}$ in Algorithm~\ref{alg:sketch} in Section~\ref{ss:approx-sketch}.
Next, we discuss the formal definition of the minimum matching between sketches in Section~\ref{ss:matching}.
Finally, we show that the minimum matching is a good approximation to $d(G,H)$
and we give a tester for isomorphism in Section~\ref{ss:approx-rooted-forest}.

\subsection{Approximation algorithms for sketches}
\label{ss:approx-sketch}

In this subsection, we analyze the behavior of the algorithm $\widetilde{\Sketch}$.
Upon computing the sketch, it is desired that the procedure~$\widetilde{\Size}_{G,\qsize,R}(T)$ always outputs a good approximation to $|V(T)|$ for an $s$-rooted tree $T$.
For $\delta' \in (0,1)$ and an $s$-rooted tree $T$,
we say that (the output of) $\widetilde{\Size}_{G,\qsize,R}(T)$ \emph{is $\delta'$-safe}
if $|\widetilde\Size_{G,q, R}(\rt(T)) - |V(T)| | \le \delta' |V(T)|$ holds.
Further, we say that (the execution of) $\widetilde{\Sketch}$ is \emph{$\delta'$-safe}
if all the outputs of $\widetilde{\Size}$ in Line~\ref{alg-line:approx-inverse-size} are $\delta'$-safe.
Let $\Sketch^{\delta'}_{\qloop, \qfreq, \qsize ,R, k}(G)(\bfu) = \E[\widetilde{\Sketch}_{\qloop, \qfreq, \qsize, R, k}(G)(\bfu) \mid \widetilde{\Sketch} : \delta'\mbox{-safe}]$.
Note that 
$\widetilde{\Sketch}$ is $\delta'$-safe with high probability
if the parameters are chosen appropriately
by Lemma~\ref{lemma:approximate-size}.
Let $G^{(i)}\ (i \in [c(G)])$ be the $i$-th $s$-rooted tree in $G$
and $v^{(i)} = \rt(G^{(i)})$.
We use the following two lemmas in the next subsection.

\begin{lemma}
  \label{lemma:approximate-count}
  For any $k,s,R,B \ge 1$ and $\delta,\delta',\tau \in (0,1)$, there exist
  $\qloop = \qIVloop(k,s, \delta,\delta', \tau)$
  and
  $\qsize = \qIVsize(k,s, \delta,\delta', \tau)$
  such that
  for any $\qfreq$ and an $R$-good $s$-rooted forest $G$ with root degrees in $(B, \gamma B]$,
  $|\widetilde\Sketch_{\qloop,\qfreq,\qsize,R,k}(G)(\bfu) - \Sketch^{\delta'}_{\qloop,\qfreq,\qsize,R,k}(G)(\bfu)| \le \frac{\delta n}{B}$
  holds for all $\bfu \in \bbN_{<k}^{t(s)}$
  with probability at least $1-\tau$.
  Here $\qIVloop$ and $\qIVsize$ are polynomials in $k^{t(s)}, \delta, \delta', \tau$.
\end{lemma}
\begin{proof}
  For simplicity, we omit the subscript of procedures.
  Let $p_{i, \bfu}$ be the probability that a vertex of $G^{(i)}$ is chosen in Line~\ref{alg-line:chose-u} and $\bfu$ is obtained in Line~\ref{alg-line:rounding} of $\widetilde{\Sketch}$.
  Denote by $\Phi$ the mapping in the algorithm~$\widetilde{\Sketch}$.

  For $\delta' \in (0,1)$, it holds that
  \begin{eqnarray*}
    \Var[\Phi(\bfu) \mid \widetilde{\Sketch} : \delta'\mbox{-safe}] &\le&
    \sum_{i \in [c(G)]} p_{i,\bfu}
      \cdot \E\left[ \frac{1}{\widetilde{\Size}_{G,\qsize,R}(v^{(i)})^2}  \middle|  \widetilde{\Size} : \delta'\mbox{-safe} \right] \\
    &\le& \sum_{i \in [c(G)]} \frac{|V(G^{(i)})|}{n}
      \cdot \left( \frac{1}{(1-\delta')|V(G^{(i)})|} \right)^2.
  \end{eqnarray*}

  Further, since we assume that the root degree of each $s$-tree in $G$ is in $(B, \gamma B]$,
  $c(G) \le n / B$ holds.
  By Chebyshev's inequality,
  \begin{eqnarray*}
    & &
    \Pr\left[ |\widetilde{\Sketch}(\bfu) - \Sketch^{\delta'}(\bfu)| \ge \frac{\delta n}{B}
      \middle| \widetilde{\Sketch}:\delta'\mbox{-safe} \right] \\
    &\le&
      \left(\frac{B}{\delta n}\right)^2
      \cdot \left(\frac{\tilde{n}}{\qloop}\right)^2
      \cdot \Var[\Phi(\bfu) \mid \widetilde{\Sketch} : \delta'\mbox{-safe}] \\
    &\le&
      \left(\frac{B}{\delta n}\right)^2
      \cdot \left(\frac{\tilde{n}}{\qloop}\right)^2
      \cdot \qloop
      \sum_{i \in [c(G)]} \frac{|V(G^{(i)})|}{n}
      \cdot \left( \frac{1}{(1-\delta')|V(G^{(i)})|} \right)^2 \\
    &\le&
      \frac{1}{\qloop\delta^2 (1-\delta')^2}
      \sum_{i \in [c(G)]} \frac{B^2}{n|V(G^{(i)})|} \\
    &\le&
      \frac{1}{\qloop\delta^2 (1-\delta')^2}.
  \end{eqnarray*}

  Here, in the last inequality, note that
  \[
    \sum_{i \in [c(G)]} \frac{B^2}{n|V(G^{(i)})|} \le \sum_{i \in [c(G)]} \frac{B}{n} = \frac{c(G)B}{n} \le 1.
  \]

  Set $\qIVloop(s,k,\delta,\delta',\tau) = O(k^{t(s)}/(\delta^2 (1-\delta')^2 \tau))$
  and $\qIVsize(s,k,\delta,\delta',\tau) = \qII(s, \delta', \tau / (2\qloop))$.
  Then, the execution of $\widetilde{\Sketch}$ is $\delta'$-safe with probability $1-\tau/(2k^{t(s)})$.
  Therefore, concerning the conditional probability,
  $|\widetilde{\Sketch}(\bfu) - \Sketch^{\delta'}(\bfu)| < \frac{\delta n}{B}$
  holds with probability $1-\tau/k^{t(s)}$.
  Applying the union bound for all $\bfu \in \bbN_{<k}^{t(s)}$, the lemma follows.
\end{proof}

\begin{lemma}
  \label{lemma:norm-of-count}
  For any $s \ge 1$ and $\delta \in (0,1)$,
  there exist linear functions $\delta'=\deltaV(\delta)$ and $\eta = \etaV(\delta)$ such that
  if $\frac{\widetilde{n}}{n} \in [1-\eta,1]$,
  then for any $R,k,B,\gamma,\qloop,\qfreq,\qsize \ge 1$ and $R$-good $s$-rooted forest $G$ with root degrees in $(B, \gamma B]$,
  $|\|\Sketch^{\delta'}_{\qloop,\qfreq,\qsize,R,k}(G)\|_1 - c(G)| \le \delta c(G)$ holds.
\end{lemma}
\begin{proof}
  Again, for simplicity, we omit the subscript of procedures.
  For $u\in V(G)$, let $v(u)$ be a root vertex of an $s$-rooted tree that $u$ belongs to.
  \begin{eqnarray*}
    \|\Sketch^{\delta'}(G)\|_1
    &=& \frac{\tilde{n}}{\qloop} \cdot \qloop
      \cdot \E_{u \in V(G)} \left[\frac{1}{\widetilde{\Size}_{G, \qsize, R}(v(u))}
      \middle| \widetilde{\Size}:\delta'\mbox{-safe} \right] \\
    &=& \tilde{n} \sum_{i\in[c(G)]}
      \Pr[\mbox{vertex of } G^{(i)} \mbox{ is chosen}] \cdot
      \E\left[ \frac{1}{\widetilde{\Size}_{G, \qsize, R}(v^{(i)})}
      \middle| \widetilde{\Size}:\delta'\mbox{-safe} \right] \\
    &=& \tilde{n} \sum_{i\in[c(G)]}
      \frac{|V(G^{(i)})|}{n} \cdot
      \E\left[ \frac{1}{\widetilde{\Size}_{G, \qsize, R}(v^{(i)})}
      \middle| \widetilde{\Size}:\delta'\mbox{-safe} \right].
  \end{eqnarray*}
  By the condition of the lemma, $\tilde{n}/n \in [1-\eta,1]$.
  Further, the expectation in the last equation is
  between $1/( (1+\delta')|V(G^{(i)})|)$ and $1/( (1-\delta')|V(G^{(i)})|)$.
  Therefore,
  $\|\Sketch^{\delta'}(G)\|_1 \in [(1-\eta)c(G)/(1+\delta'), c(G)/(1-\delta')]$.
  Setting $\etaV(\delta) = \deltaV(\delta) = \delta/2$,
  the lemma follows.
\end{proof}

Hereafter,
for $\delta,\tau \in (0,1)$ and $\qfreq, k, R \ge 1$,
we denote by $\widetilde{\Sketch}_{\delta, \tau, \qfreq, R, k}$
the procedure $\widetilde{\Sketch}_{\qloop, \qfreq, \qsize, R, k}$
for $\delta'=\deltaV(\delta)$, $\qloop=\qIVloop(k,s, \delta,\delta', \tau)$, and $\qsize=\qIVsize(k,s, \delta,\delta', \tau)$
in order to simplify the notion.
Further, we denote by $\Sketch_{\delta, \tau, \qfreq, R, k}$
the conditional expectation $\Sketch^{\delta'}_{\qloop, \qfreq, \qsize, R, k}$.

Finally, we consider the (expected) query complexity of $\widetilde{\Sketch}$.
Since $\qIVloop, \qIVsize$ are polynomials in $k$, $\delta$, and $\tau$, the following holds.

\begin{lemma}
  \label{lemma:query-complexity-of-sketch}
  For $\delta,\tau \in (0,1)$ and $R,k \ge 1$,
  the expected query complexity of the algorithm~$\widetilde{\Sketch}_{\delta, \tau, \qfreq, R, k}$
  is a polynomial in $\delta, \tau, \qfreq, R, k^{t(s)}$.
  \qed
\end{lemma}


\subsection{Matching sketches}
\label{ss:matching}
In this subsection, we define the distance between two sketches so that it is a good approximation to $\dist(G,H)$.
Let denote by $\PReal$ the set of non-negative reals.

A \emph{weighted point set} is a tuple $X = (w,S)$,
where $w:S \to \PReal$ is a \emph{weight function} and $S$ is a set of vectors.
To define the distance between sketches, we consider the following problem.
\begin{definition}[Minimum matching between weighted point sets]
  Let $X_1 = (w_1, S_1)$ and $X_2 = (w_2, S_2)$ be two weighted points sets with $\|w_1\|_1 = \|w_2\|_1$.
  We call a function $f:S_1 \times S_2 \rightarrow \PReal$
  a flow function from $X_1$ to $X_2$ if
  $\sum_{\bfv \in S_2} f(\bfu, \bfv) = w_1(\bfu)$ for all $\bfu \in S_1$ and
  $\sum_{\bfu \in S_1} f(\bfu, \bfv) = w_2(\bfv)$ for all $\bfv \in S_2$.
  The \emph{value} of a flow function $f$ is defined as
  \[
    \sum_{(\bfu,\bfv) \in S_1 \times S_2} f(\bfu, \bfv) \cdot \|\bfu - \bfv\|_1.
  \]
  The minimum value of a flow function is denoted by $\MM(X_1, X_2)$,
  and the flow function that achieves the minimum value is called the \emph{optimal flow function}.
\end{definition}

Note that the optimal flow function can be calculated by a min-cost flow algorithm on a bipartite graph.
Therefore, the following lemma holds.
\begin{lemma}  \label{cor:solution-integral}
  Let $X_1 = (w_1,S_1)$ and $X_2 = (w_2,S_2)$ be weighted point sets.
  If $w_1$ and $w_2$ are integral,
  there exists an optimal flow function $f^*$ that is integral.
  In particular, if all values of $w_1$ and $w_2$ are $1$,
  the set of pairs $\{(\bfu, \bfv) \in S_1 \times S_2 \mid f^\ast(\bfu,\bfv)=1\}$ forms a matching.
  \qed
\end{lemma}

For an $s$-rooted forest $G$, let $\Freq(G)$ denote the multiset $\{\Freq(G^{(1)}), \cdots, \Freq(G^{(c(G))})\}$.
To define the distance between sketches, we first associate weighted point sets $F_G$, $S_{G}$, and $\tilde{S}_{G}$ with $\Freq(G)$, $\Sketch(G)$, and $\widetilde\Sketch(G)$, respectively.
Then, we show that $\MM(F_G,F_H)$ can be well approximated by $\MM(\tilde{S}_{G},\tilde{S}_{H})$.
Next, we show that $\dist(G, H)$ can be approximated by $\MM(F_G,F_H)$.
Since we can efficiently compute $\tilde{S}_{G}$ and $\tilde{S}_{H}$,
it follows that we can well approximate $\dist(G,H)$.
Hence, we can test isomorphism between $G$ and $H$.

We first introduce auxiliary weighted point sets.
Let $F'_G = (\bfone, \Freq(G))$, where $\bfone$ is the constant-one function.
For parameters $\delta, \tau, \qfreq, R, k$, we define $S'_{G,\delta,\tau,\qfreq,R,k}$ as follows.
First for a cell $C$, we define $\vtx(C)$ as the unique point in $C$ that is minimal with respect to every axis.
Then, for each $\bfu \in \bbN_{<k}^{t(s)}$, 
we add a point $\vtx(\Cell(\bfu))$ with weight $\Sketch_{\delta,\tau,\qfreq,R,k}(G)(\bfu)$.
Similarly, we define $\tilde{S}'_{G,\delta,\tau,\qfreq,R,k}$.
If the parameters are clear from the context, we occasionally drop the subscripts of $\tilde{S}'$ and $S'$.
A technical issue here is that the sums of weights of $F'_G$, $S'_{G}$, and $\tilde{S}'_{G}$ might be different since $\widetilde\Sketch$ is a random variable,
and it means that we cannot define matchings among them.
To avoid this issue, 
for a large integer value $M$, 
we define $\ext((w',S'), M)=(w,S)$ as the extension of $(w',S')$
so that $S = S' \cup \{\bot\}$ and $w(\bot)=M-\|w'\|_1$.
We regard $\bot$ as the all-zero vector when measuring distances to other vectors.
For a sufficiently large $M$,
we define $F_G = \ext(F'_G, M)$, $S_{G} = \ext(S'_{G}, M)$, and $\tilde{S}_{G} = \ext(\tilde{S'}_{G}, M)$.

This section is devoted to prove the following lemma.

\begin{lemma}
  \label{lemma:mmgh-mmtcgtch}
  For any $s, R, \gamma \ge 1$, and $\delta'', \tau' \in (0,1)$,
  there exist parameters $\delta$, $\tau$, $\qfreq$, and $k$
  such that
  $|\MM(F_{G}, F_{H}) - \MM(\tilde{S}_{G}, \tilde{S}_{H})| \le \delta'' n$
  holds with probability at least $1-\tau'$.
  The parameters $\delta, k$ are polynomials in $\gamma, \delta''$,
  $\tau$ is $O(\tau')$,
  and $\qfreq$ is a polynomial in $\gamma, \delta'', R$.
\end{lemma}

To prove Lemma~\ref{lemma:mmgh-mmtcgtch}, we prove several lemmas first.
\begin{lemma}
  \label{lemma:count-tildecount-diff}
  For any $s$,$\gamma$,$\qfreq \ge 1$, and $\delta'', \tau \in (0,1)$,
  there exists $\delta = \deltaVI(s, \gamma, \delta'')$ such that
  \[
  \MM(S_{G, \delta, \tau, \qfreq, R, k}, \tilde{S}_{G, \delta, \tau, \qfreq, R, k}) \le \delta'' n
  \]
  with probability at least $1-\tau$.
  Here $\deltaVI$ is a polynomial in $\gamma, \delta''$.
\end{lemma}
\begin{proof}
  We construct a flow function from $S_{G}$ to $\tilde{S}_{G}$
  so that
  $\sum_{(\bfu, \bfv)} f(\bfu, \bfv)  \|\bfu - \bfv\|_1 \le \delta'' n$ holds with high probability.
  We assign $f(\bfu, \bfu) = \min(\Sketch(G)(\bfu),\widetilde{\Sketch}(G)(\bfu))$ for each $\bfu \in \bbN_{<k}^{t(s)}$
  and assign an arbitrary value to other parts of $f$
  so that $f$ satisfies the condition of a flow function.
  By Lemma~\ref{lemma:approximate-count},
  $\sum_{\bfu,\bfv \in \bbN_{<k}^{t(s)}, \bfu \neq \bfv} f(\bfu, \bfv) \le \frac{\delta n}{B}$
  with probability at least $1-\tau$.
  Set $\delta = \deltaVI(s, \gamma, \delta'') = \delta'' / (t(s) \gamma)$.
  Since $\|\bfu - \bfv\|_1 \le t(s) \gamma B$, we have
  \[
    \sum_{(\bfu, \bfv)} f(\bfu, \bfv)\|\bfu-\bfv\|_1
    \le \frac{\delta n}{B} \cdot t(s)  \gamma B
    \le \delta'' n.
  \]
\end{proof}

\begin{lemma}
  \label{lemma:g-count-diff}
  For any $s, R, \gamma \ge 1$ and $\delta'', \tau \in (0,1)$,
  there exist $k = \kI(s, \gamma, \delta'')$,
  $\qfreq=\qVIIfreq(s, \gamma, \delta'', R)$,
  $\delta = \deltaVII(s, \gamma, \delta'')$
  such that $\MM(F_{G}, S_{G,\delta,\tau,\qfreq,R,k}) \le \delta'' n$ holds.
  The parameters $\kI,\deltaVII$ are polynomials in $\gamma, \delta''$
  and $\qVIIfreq$ is a polynomial in $\gamma, \delta'', R$.
\end{lemma}
\begin{proof}
  Let $k, \delta, \qfreq$ be parameters chosen later.
  Again, we construct a flow function $f$ from $F_{G}$ to $S_{G}$.
  We define a hypercube $C_i$ in $\Real^{t(s)}$ as $C_i = \{(x_1, x_2, \ldots, x_{t(s)}) \mid |x_j - \Freq(G^{(i)})[j]| <  \gamma B / (2k), j \in [t(s)] \}$.
  Further, let $B_i = \{\vtx(\Cell(\bfv)) \mid \bfv \in \bbN_{<k}^{t(s)}, \Cell(\bfv) \cap C_i \neq \emptyset\}$.
  Note that $|B_i| \le 2^{t(s)}$.

  Let $\delta' = \deltaV(\delta)$.
  Let $p^{(i)}$ be the probability that a vertex of $G^{(i)}$ is chosen in Line~\ref{alg-line:chose-u} of the algorithm $\widetilde{\Sketch}$,
  $q_{i, \bfv}$ be the probability that $\Round(\widetilde{\Freq}_{\qfreq}(v^{(i)})) = \bfv$ holds,
  and $e^{(i)} = \E\left[ 1/\widetilde{\Size}(v^{(i)}) \middle| \widetilde{\Size}:\delta'\mbox{-good} \right]$.
  For $i \in [c(G)]$ and $\bfv \in \bbN_{<k}^{t(s)}$,
  define a flow function as follows.
  \[
    f(i, \bfv) = (1-\delta')\tilde{n} p^{(i)} q_{i, \bfv} e^{(i)}
  \]

  We show that $\sum_{i} f(i, \bfv) \le \Sketch(G)(\bfv)$ for all $\bfv$ and $\sum_{\bfv} f(i, \bfv) \le 1$ for all $i$.
  The conditional expectation $\Sketch$ can be expressed as $\Sketch(G)(\bfv) = \sum_{i} \tilde{n} p^{(i)} q_{i, \bfv}  e^{(i)}$.
  Thus, $\sum_{i} f(i, \bfv) = (1-\delta')\Sketch(G)(\bfv) < \Sketch(G)(\bfv)$ holds.
  Further, since $p^{(i)} e^{(i)} \le (|V(G^{(i)})| / n) \cdot 1/((1-\delta')|V(G^{(i)})|) = 1/(n(1-\delta'))$
  and $\sum_{\bfv} q_{i, \bfv} = 1$, $\sum_{\bfv} f(i, \bfv) \le 1$ holds.
  Similarly, we can show that 
  $\sum_{\bfv} f(i, \bfv) \ge (1-\delta')/(1+\delta') \ge 1 - 2\delta'$ for all $i$.

  We assign values to the remaining part of $f$ so that the condition of the flow function is satisfied.
  Here it holds that
  $\sum_{\bfv} f(\bot, \bfv)
  = \sum_{\bfv} (\Sketch(G)(\bfv) - \sum_{i} f(i, \bfv))
  = \delta' \| \Sketch(G) \|_1$
  and 
  $\sum_{i} f(i, \bot)
  = \sum_{i} (1 - \sum_{\bfv} f(i, \bfv))
  \le 2\delta'c(G)$.
  From Lemma~\ref{lemma:norm-of-count},
  we have 
  $\sum_{\bfv} f(\bot, \bfv) + \sum_{i} f(i, \bot) \le 4\delta' c(G)$.

  Let $r_i = \sum_{\bfv \in B_i} q_{i,\bfv}$.
  Set $k = \kI(s, \gamma, \delta'') = O(2^{t(s)}t(s)\gamma / \delta'')$
  and $\qfreq = \qVIIfreq(s, \gamma, \delta'', R) = \qI(s, 1/(2kR), \tau')$
  for $\tau' = O(1 / (k^{t(s)} \cdot t(s) \gamma))$. Then from Lemma~\ref{lemma:approximate-freq}, we have
  $r_i \ge 1-\tau'$.

  Now, we calculate the value of the flow function. For fixed $i \in [c(G)]$,
  \begin{eqnarray*}
    \sum_{\bfv} f(i, \bfv) \cdot \|\Freq(G^{(i)}) - \bfv\|_1
    &=&
      \sum_{\bfv \in B_i} f(i, \bfv) \cdot \|\Freq(G^{(i)}) - \bfv\|_1
      + \sum_{\bfv \not\in B_i} f(i, \bfv) \cdot \|\Freq(G^{(i)}) - \bfv\|_1 \\
    &\le&
      \sum_{\bfv \in B_i} 1 \cdot \frac{t(s)\gamma B}{k}
      + \sum_{\bfv \not\in B_i} (1 - r_i) \cdot t(s)\gamma B\\
    &\le&
      2^{t(s)} \cdot 1 \cdot \frac{t(s)\gamma B}{k}
      + k^{t(s)} \tau' \cdot t(s)\gamma B\\
    &\le&
      \frac{\delta''  B}{4} + \frac{\delta''  B}{4} = \frac{\delta'' B}{2}.
  \end{eqnarray*}

  Set $\delta = \deltaVII(s, \gamma, \delta'') = O(\delta'' / (k^{t(s)}\gamma))$ so that $\delta' = O(1/(k^{t(s)} \gamma))$. Then we have
  \begin{eqnarray*}
    \sum_{i, \bfv} f(i, \bfv) \cdot \|\Freq(G^{(i)}) - \bfv\|_1 
    &\le& c(G) \cdot \frac{\delta'' B}{2}
     = \frac{\delta'' c(G) B}{2} \\
    \sum_{\bfv} f(\bot, \bfv) \| \bfv \|_1 + \sum_{i} f(\Freq(G^{(i)}), \bot) \| \Freq(G^{(i)}) \|_1
    &\le& 4\delta' c(G) \cdot k^{t(s)} \gamma B
     = \frac{\delta'' c(G) B}{2}.
  \end{eqnarray*}
  Since $c(G)B \le n$, the cost of the flow function is at most $\delta'' n$.
\end{proof}

We can show that the triangle inequality holds for the minimum value of a flow function $\MM$.
Combining Lemmas~\ref{lemma:count-tildecount-diff},~\ref{lemma:g-count-diff}, and the triangle inequality,
we can prove Lemma~\ref{lemma:mmgh-mmtcgtch}.

\begin{lemma}
  \label{lemma:mm-triangle}
  Let $X_i = (w_i, S_i)$ $(i=1,2,3)$ be weighted point sets with $\|w_1\|_1 = \|w_2\|_1 = \|w_3\|_1$.
  Then, the triangle inequality holds for $\MM(\cdot,\cdot)$ among them, that is
  \[
    \MM(X_1, X_3) \le \MM(X_1, X_2) + \MM(X_2, X_3).
  \]
\end{lemma}
\begin{proof}
  We construct a flow function $f_{13}$ from $X_1$ to $X_3$ as follows.
  Let $f^\ast_{12}$ be the optimal flow function from $X_1$ to $X_2$,
  and let $f^\ast_{23}$ be the one from $X_2$ to $X_3$.
  Let $f_{13}(i_1, i_3) = \sum_{i_2} \frac{f^\ast_{12}(i_1, i_2) f^\ast_{23}(i_2, i_3)}{w_2(i_2)}$.
  The function $f_{13}$ satisfies the conditions of a flow function:
  \begin{eqnarray*}
    \sum_{i_3} f_{13}(i_1, i_3) 
    = \sum_{i_3} \sum_{i_2}  \frac{f^\ast_{12}(i_1, i_2) f^\ast_{23}(i_2, i_3) }{w_2(i_2)} = \sum_{i_2} f^\ast_{12}(i_1, i_2) = w_1(i_1),\\
    \sum_{i_1} f_{13}(i_1, i_3) 
    = \sum_{i_1} \sum_{i_2}  \frac{f^\ast_{12}(i_1, i_2) f^\ast_{23}(i_2, i_3)}{w_2(i_2)}
    = \sum_{i_2} f^\ast_{23}(i_2, i_3) = w_3(i_3).
  \end{eqnarray*}

  We observe that
  \begin{eqnarray*}
    & & \sum_{i_1,i_3} f_{13}(i_1, i_3)\|S_1(i_1) - S_3(i_3) \|_1 \\
    &=&   \sum_{i_1, i_3} \sum_{i_2} \frac{f^\ast_{12}(i_1, i_2) f^\ast_{23}(i_2, i_3) \|S_1(i_1) - S_3(i_3)\|_1}{w_2(i_2)} \\
    &\le& \sum_{i_1, i_3} \sum_{i_2} \frac{f^\ast_{12}(i_1, i_2) f^\ast_{23}(i_2, i_3) (\|S_1(i_1) - S_2(i_2)\|_1 + \|S_2(i_2) - S_3(i_3)\|_1)}{w_2(i_2)} \\
    &=&  \sum_{i_1, i_2} f^\ast_{12}(i_1, i_2) \|S_1(i_1) - S_2(i_2)\|_1 + \sum_{i_2, i_3} f^\ast_{23}(i_2, i_3)\|S_2(i_2) - S_3(i_3) \|_1. \\
  \end{eqnarray*}
  Thus, $\MM(X_1, X_3) \le \MM(X_1, X_2) + \MM(X_2, X_3)$.
\end{proof}

\begin{proof}[Proof of Lemma~\ref{lemma:mmgh-mmtcgtch}]
  Set
  $\qfreq=\qVIIfreq(s, \gamma, O(\delta''), R)$, \\
  $\delta = \min(\deltaVI(s, \gamma, O(\delta'')), \deltaVII(s, \gamma, O(\delta'')))$,
  $\tau = \tau'/2$,
  and
  $k = \kI(s, \gamma, O(\delta''))$.
  Then with probability at least $1-\tau'/2$,
  \[
    \MM(F_{G}, \tilde{S}_{G})
    \le
    \MM(F_{G}, S_{G}) + \MM(S_{G}, \tilde{S}_{G})
    \le
    O(\delta'' n) + O(\delta'' n) = O(\delta'' n).
  \]
  The same inequality holds for the other graph $H$.
  Thus, with probability $1-\tau'$,
  \begin{eqnarray*}
    & &|\MM(F_{G}, F_{H}) - \MM(\tilde{S}_{G}, \tilde{S}_{H})| \\
    &\le&
      |\MM(F_{G}, F_{H}) - \MM(F_{H}, \tilde{S}_{G})| +
      |\MM(F_{H}, \tilde{S}_{G}) - \MM(\tilde{S}_{G}, \tilde{S}_{H})| \\
    &\le& \MM(F_{G}, \tilde{S}_{G}) + \MM(F_{H}, \tilde{S}_{H}) \le O(\delta'' n) + O(\delta'' n') = \delta'' \tilde{n}.
  \end{eqnarray*}
\end{proof}


\subsection{Approximation algorithm}
\label{ss:approx-rooted-forest}
Finally, we show that the distance between two graphs can be well approximated by the minimum value of a matching between corresponding sketches.
First, we need to show the following.
\begin{lemma}
  \label{lemma:mmgh-dist}
  Let $G$ and $H$ be $s$-rooted forests.
  Then, $\dist(G, H) \le 2s \cdot \MM(F_{G}, F_{H})$.
\end{lemma}
\begin{proof}
  Let $f^\ast$ be the optimal flow function achieving $\MM(F_G, F_H)$.
  By Lemma~\ref{cor:solution-integral},
  we assume that every value of $f^\ast$ is 0 or 1.
  Therefore, we regard the flow function as a matching:
  Let $F_G = (\bfone, \{\bot, \bfu_1,\ldots,\bfu_{c(G)}\})$ and $F_H = (\bfone, \{\bot, \bfv_1,\ldots,\bfv_{c(H)}\})$.
  Then, consider a bipartite graph such that the left part consists of $\{G^{(i)} \}$, the right part consists of $\{H^{(j)}\}$, and there is an edge between $G^{(i)}$ and $H^{(j)}$ iff $f^\ast(\bfu_i, \bfv_j) = 1$.
  Then, this graph forms a (partial) matching.

  Using $f^\ast$, we construct a sequence of modifications to transform $G$ to $H$.
  For each $G^{(i)}$ with $f(\bfu_i, \bot) = 1$, we remove all the edges in $G^{(i)}$.
  For each $H^{(j)}$ with $f(\bot, \bfv_j) = 1$, we remove all the edges in $H^{(j)}$.

  Consider a pair $G^{(i)}$ and $H^{(j)}$ for which $f(\bfu_i, \bfv_j) = 1$.
  Let $\calT_{G^{(i)}}$ and $\calT_{H^{(j)}}$ be the set of subtrees in $G^{(i)}$ and $H^{(j)}$, respectively.
  From the definition, we can choose $\|\Freq(G^{(i)}) - \Freq(H^{(j)})\|_1$ sets of subtrees $\calT'_{G^{(i)}} \subseteq \calT_{G^{(i)}}$ and $\calT'_{H^{(j)}} \subseteq \calT_{H^{(j)}}$ in total so that $\calT_{G^{(i)}} \setminus \calT'_{G^{(i)}}$ and $\calT_{H^{(j)}} \setminus \calT'_{H^{(j)}}$ are isomorphic.

  The total number of edge modifications is bounded by
  \begin{align*}
    & \sum_{\bfu_i:f(\bfu_i,\bot)=1} s\deg(\rt(G^{(i)})) 
    + 
    \sum_{\bfv_j:f(\bot,\bfv_j)=1} s\deg(\rt(H^{(j)})) \\
    & \qquad + \sum_{(\bfu_i,\bfv_j):f(\bfu_i,\bfv_j)=1} s\|\Freq(G^{(i)}) - \Freq(H^{(j)})\|_1 \\
    &\le 2s \cdot \MM(F_G, F_H).
  \end{align*}
\end{proof}

Now, we prove Lemma~\ref{lemma:upper-bound-for-simple-case}.
We show that the following algorithm is a tester for forest-isomorphism.

\begin{algorithm}[!h]
  \caption{tests whether $d(G,H)=0$ or $d(G, H) \ge \varepsilon' \tilde{n}$,
    with probability at least $1-\tau'$,
    given $\tilde{n},s,R,B,\gamma \ge 1$, $\varepsilon',\tau',\eta \in (0,1)$
    and $R$-good $s$-rooted forest with root degree in $(B, \gamma B]$
    with $\frac{\tilde{n}}{n}, \frac{\tilde{n}}{n'} \in [1-\eta, 1]$
    for $n = |V(G)|$ and $n' = |V(H)|$.
  }
\begin{algorithmic}[1]
  \Procedure{$\TestSimple_{\varepsilon',s,\tau',\eta,\gamma,R}(G, H, \tilde{n}, B)$}{}
    \State{Set $\delta'' = O(\varepsilon' / s)$.}
    \State{Choose parameters $\delta, \tau, \qfreq, k$ in Lemma~\ref{lemma:mmgh-mmtcgtch} according to parameters $\gamma, \delta'', \tau'/2, R$. }
    \State{Compute $\widetilde{S}_{G} = \widetilde\Sketch_{\delta, \tau, \qfreq, R, k}(G)$ and $\widetilde{S}_{H} = \widetilde\Sketch_{\delta, \tau, \qfreq, R, k}(H)$.}
    \State{Compute $\widetilde{M} = \MM(\widetilde{S}_{G}, \widetilde{S}_{H})$ by a min-cost flow algorithm.}
    \If{$\widetilde{M} < \delta'' \tilde{n}$}
      \State{\ret YES}
    \Else
      \State{\ret NO}
    \EndIf
  \EndProcedure
\end{algorithmic}
\end{algorithm}

\begin{lemma}[Restatement of Lemma~\ref{lemma:upper-bound-for-simple-case}]
  \label{lemma:test-simple}
  There exists $\eta = \etaI(s, \varepsilon', \gamma)$
  such that for any input with $\frac{\tilde{n}}{n}, \frac{\tilde{n}}{n'} \in [1-\eta, 1]$,
  the procedure~$\TestSimple$ correctly decides $d(G,H) = 0$ or $d(G,H) \geq \varepsilon' \tilde{n}$ with probability at least $1-\tau'$.
  Here $\eta$ is a polynomial in $\varepsilon', \gamma$.
  Regarding that $s$ is constant, the query complexity is a polynomial in $\gamma, \varepsilon', \tau', R$.
  Denote by $\qrandom(s, \gamma, \varepsilon', \tau')$ the number of random vertex queries the procedure invokes.
  Then $\qrandom$ is a polynomial in $\gamma, \varepsilon', \tau'$.
\end{lemma}
\begin{proof}
  Set $\eta = \etaV(\delta)$.
  Combining Lemmas~\ref{lemma:mmgh-mmtcgtch} and \ref{lemma:mmgh-dist},
  the correctness of $\TestSimple$ can be proven as follows:
  If $d(G, H) = 0$, $\widetilde{M} \le \MM(F_{G}, F_{H}) + \delta''\tilde{n} = \delta''\tilde{n}$ (with probability $1-\tau'$).
  On the other hand, 
  if $d(G, H) \ge \varepsilon' \tilde{n}$,
  $\widetilde{M} \ge \MM(F_{G}, F_{H}) - \delta''\tilde{n} \ge (\varepsilon'/(2s) - \delta'') \tilde{n} > \delta'' \tilde{n}$. 

  The query complexity of $\TestSimple$ is polynomial in $\delta, \tau, \qfreq, R, k^{t(s)}$.
  Since parameters $\delta, \tau, k$ are polynomials in $\gamma, \delta'' = O(\varepsilon' / s), \tau'$,
  and $\qfreq$ is a polynomial in $\gamma, \delta'', \tau', R$,
  the query complexity is a polynomial in $\gamma, \varepsilon', \tau', R$.
  We invoke random vertex queries $O(\qloop)$ times for $\qloop = \qIVloop(k,s, \delta, \deltaV(\delta), \tau)$,
  and therefore $\qrandom$ is a polynomial in $\gamma, \varepsilon', \tau'$.
\end{proof}

\section{Missing Parts of Section~\ref{sec:general-case}}
\label{appendix:missing-general-case}

\subsection{Missing proofs from Section~\ref{sec:general-case}}
\label{ss:missing-general-case}

\begin{proof}[of Lemma~\ref{lemma:boundary-bound}]
  Note that $\frac{1+\mu}{1-\mu} < \gamma$.
  For a tree $T$ in $G$, let $p_T$ be the probability that $T$ is on the $(\alpha, \gamma, \mu)$-boundary and $d'_T = \deg(\rt(T))$.
  Note that $T$ is on the $(\alpha, \gamma, \mu)$-boundary if and only if
  $\alpha \in [\frac{\gamma^i}{d'_T(1+\mu)}, \frac{\gamma^i}{d'_T(1-\mu)}]$ for some $i$.
  Let $f_T(x) := |[1,\gamma] \cap [\frac{x}{d'_T(1+\mu)}, \frac{x}{d'_T(1-\mu)}]|$.
  Then, $p_T = \sum_{i \ge 1} f_T(\gamma^i) / (\gamma-1)$ since the intervals $\{[\frac{\gamma^i}{d'_T(1+\mu)},
  \frac{\gamma^i}{d'_T(1-\mu)}]\}_{i \ge 1}$ are disjoint as $\frac{1+\mu}{1-\mu} < \gamma$.

  From the definition, if $\frac{x}{d'_T(1-\mu)} \le 1$ or $\frac{x}{d'_T(1+\mu)}\ge \gamma$, then $f_T(x)=0$ and otherwise $f_T(x)>0$.
  Futher, if $f_T(x) > 0$, then $f_T(\gamma^2 x)=0$ since this implies $x > d'_T(1-\mu)$ and $\gamma^2 x > d'_T(1-\mu)\gamma \cdot \gamma > d'_T(1+\mu)\gamma$.
  It follows that $\#\{i \in \bbN_{< L+1} \mid f_T(\gamma^i) > 0\} \le 2$.
  The value of $f_T(x)$ is maximized when $\frac{x}{d'_T(1-\mu)} = \gamma$.
  Thus, $f_T(x) \le \gamma - \frac{1-\mu}{1+\mu}\gamma \le 2\gamma\mu$, and we have $p_T \le 4\gamma\mu$.

  Therefore, we have $\E_\alpha[B_{\alpha,\gamma,\mu}(G)] = \sum_{T} p_T |V(T)| \le 4\gamma\mu n$.
  By Markov's inequality, the lemma holds.
\end{proof}

\begin{proof}[of Lemma~\ref{lemma:which-component}]
  For a vertex $v \in V(\GP{0}) \cup \cdots \cup V(\GP{L})$, let $T$ be a tree with $v \in T$.
  If $T$ contains no high-degree vertex, the procedure $\Which$ decides that $v \in \GP{0}$. This output is correct.

  Suppose that $T$ contains a high-degree vertex $u$.
  Set $\delta = O(\mu / (\gamma^2 R))$ and $q = \qW(\gamma, \mu, R, \tau) = \qD(\delta, \tau)$.
  From Lemma~\ref{lemma:approx-degree} and the definition of an $R$-good tree,
  $|\widetilde{\deg}_q(u) - \deg(u)| \le \delta R \deg(u)$ holds with probability $1 - \tau$.
  If $v \in \GP{0}$ or $v \in \GP{1}$, the output is correct since $\delta R \cdot \alpha \gamma < 1/2$.
  Suppose that $v \in \GP{i}$ for $i \in [2, L]$.
  Since $\GP{i}$ is the union of trees that are not on the $(\alpha, \gamma, \mu)$-boundary,
  $(1+\mu)\alpha\gamma^i < \deg(u) < (1-\mu)\alpha\gamma^{i+1}$ holds.
  Thus, we obtain $\alpha\gamma^i < \widetilde{\deg}(u) < \alpha\gamma^{i+1}$.
\end{proof}

\begin{proof}[of Lemma~\ref{lemma:size-of-garbage}]
  The lemma follows from Lemmas~\ref{lemma:R-bad}~and~\ref{lemma:boundary-bound}.
\end{proof}

\begin{proof}[of Lemma~\ref{lemma:provide-random-access}]
  Let $A_k$ be the event that
  when running the procedure $\Random_q(G, i)$,
  we pick up vertices of $(V(\GP{0}) \cup \cdots \cup V(\GP{L})) \setminus V(\GP{i})$ $k$ times
  and pick up a vertex of $V(\GP{i})$
  and then, the procedure $\Which_q$ outputs always the correct value.
  Set $t = O(\log(1/\tau)/\delta)$.
  It is sufficient to show that $\Pr[A_0 \vee \cdots \vee A_t] \ge 1-\tau$ holds by appropriately choosing parameters.
  Let $\tau' = O(\delta \tau)$, $\lambda = \lambdaR(\delta, \tau) = O(\delta\tau)$, $q = \qW(\gamma, \mu, R, \tau')$. Then,
  \begin{eqnarray*}
    \Pr\left[ A_0 \vee \cdots \vee A_t \right]
    &=& \sum_{k=0}^{t} ( (1-\lambda-\delta)(1-\tau') )^k \cdot \delta(1-\tau') \\
    &\ge& \delta(1-\tau') \sum_{k=0}^{t} (1-\lambda-\delta-\tau')^k  \\
    &\ge& \delta(1-\tau') \cdot (1-\tau')/(\lambda+\delta+\tau') \ge 1-\tau.
  \end{eqnarray*}
\end{proof}

\subsection{Tester for isomorphism of $s$-bounded-degree forests}
\label{ss:s-bounded-degree}

We consider a forest-isomorphism tester for $s$-bounded-degree forests.
As mentioned in Section~\ref{sec:simple-case},
we need to make a tester for two forests containing different number of vertices.
Using the result in \cite{Newman:2013hg}, we can construct such a tester.
\begin{lemma}
  \label{lemma:s-bounded-degree}
  There exists a procedure such that the following holds:
  For any $\varepsilon', \tau \in (0,1)$ and $s \ge 1$,
  there exists $\eta = \etaII(\varepsilon')$
  such that for any $s$-bounded-degree forests $G$ and $H$ with $\frac{\tilde{n}}{n}, \frac{\tilde{n}}{n'} \in [1-\eta, 1]$,
  where $n = |V(G)|$ and $n' = |V(H)|$,
  the procedure correctly decides $d(G, H) = 0$ or $d(G,H) \geq \varepsilon' \tilde{n}$ with probability at least $1-\tau$.
  The query complexity depends only on $\varepsilon'$ and $\tau$.
\end{lemma}

\begin{proof}
  We use the similar notion in \cite{Newman:2013hg}.
  For $s, k \ge 1$, let $N(s, k)$ be the number of rooted graphs whose degree is at most $s$ and radius is at most $k$.
  Suppose that the $N(s, k)$ rooted graphs are numbered from 1. Let $N_i(s, k)$ be the $i$-th graph.
  In addition, for a vertex $v$ in $G$,
  let $B_G(v, k)$ be the subgraph rooted at $v$ that is induced by all vertices of $G$ that are at distance at most $k$ from $v$.
  Let $\Dist_k(G)$ be the $N(s, k)$-dimensional vector
  whose $i$-th element is the number of vertices $v$ in $G$ such that
  $B_G(v, k)$ is isomorphic to $N_i(s, k)$.
  Let $\Freq_k(G) = \Dist_k(G) / n$.
  The main result of \cite{Newman:2013hg} is as follows.

  \begin{theorem}[\cite{Newman:2013hg}]
    \label{theorem:newman2011}
    For any $\varepsilon' \in (0, 1)$ and $s \ge 1$,
    there exists $D = \DII(\varepsilon', s)$ and $\delta = \deltaII(\varepsilon', s)$ such that
    if two forests $G$ and $H$ containing equal vertices are $\varepsilon'$-far from isomorphic, then
    $\|\Freq_{D}(G) - \Freq_D(H)\|_1 > \delta''$ holds.
    \qed
  \end{theorem}

  Without loss of generality, assume that that $n \ge n'$.
  Let $H'$ be a graph consisting of $H$ and $(n-n')$ isolated vertices.
  We will prove the following claim.

  \begin{claim}
    \label{claim:hp-h}
    If $\frac{\tilde{n}}{n}, \frac{\tilde{n}}{n'} \in [1-\eta, 1]$,
    $\|\Freq_D(H') - \Freq_D(H)\|_1 \le 4\eta$.
  \end{claim}

  Using the claim, we can prove the lemma as follows.
  Set $\eta = O(\lambda)$.
  From the triangle inequality,
  $\|\Freq_D(G) - \Freq_D(H)\|_1 \ge \|\Freq_D(G) - \Freq_D(H')\|_1 - \|\Freq_D(H') - \Freq_D(H)\|_1$.
  If $G$ and $H$ are $\varepsilon'$-far from isomorphic, then
  $\|\Freq_D(G) - \Freq_D(H)\|_1 \ge \lambda - O(\lambda) = \lambda/2$ holds
  by Theorem~\ref{theorem:newman2011} and the above claim.
  Since we can approximate $\Freq$ by randomly sampling vertices and performing the BFS,
  we can test whether $\Freq_D(G) = \Freq_D(H)$ or $\|\Freq_D(G) - \Freq_D(H)\|_1 \ge \lambda / 2$ with high probability.

  We prove the claim.
  By the condition, $|1 - n/n'| \le 2\eta$.
  Suppose that an isolated vertex is indexed one in the $N(s, D)$-dimensional vector.
  Let $z$ be the number of isolated vertices in $H$. Then,
  \begin{eqnarray*}
    & & \|\Freq_D(H') - \Freq_D(H)\|_1 \\
    &=& \|\Dist_D(H')/n - \Dist_D(H)/n' \|_1 \\
    &=& |(z+n-n')/n - z/n'| + \sum_{2 \le i \le N(s, k)} |\Dist_D(H)[i]/n - \Dist_D(H)[i]/n'| \\
    &\le& |1-n'/n| + z|1/n - z/n'| + (\|\Dist_D(H)\|_1 - z)|1/n - 1/n'| \\
    &\le& |1-n'/n| + |1-n'/n| \le 4\eta.
  \end{eqnarray*}
\end{proof}

We denote the procedure in Lemma~\ref{lemma:s-bounded-degree} by
$\TestBounded_{\varepsilon', s, \tau, \eta}(G, H)$.

\subsection{Proof of Theorem~\ref{thm:forest-isomorphism}}
\label{ss:proof-of-isomorphism-tester}

Now, we prove Theorem~\ref{thm:forest-isomorphism}.
We show a tester for forest-isomorphism for $s$-forests in Algorithm~\ref{algorithm:test-isomorhism}.
Let $\TestEach_{\varepsilon', s, \tau, \eta, \gamma, R}(i, G, H, \tilde{n})$ be the algorithm
that runs $\TestBounded_{\varepsilon', s, \tau, \eta}(G, H)$ if $i=0$,
and runs $\TestSimple_{\varepsilon',s,\tau,\eta,\gamma,R}(G, H, \tilde{n}, \alpha \gamma^{i})$ otherwise.

We use the procedure $\TestIsomorphism$ in Algorithm~\ref{algorithm:test-isomorhism} as a tester.
It is sufficient to show the following theorem.

\begin{algorithm}[!h]
  \caption{returns YES if $d(G,H)=0$ and NO if $d(G,H) \ge \varepsilon n$ with high probability,
  given two $s$-forests $G,H$ and $\varepsilon,\tau \in (0,1)$.}
  \label{algorithm:test-isomorhism}
\begin{algorithmic}[1]
  \Procedure{${\TestIsomorphism}_{\varepsilon, \tau}(G, H)$}{}
    \State{Let $\gamma = 2s$, $L = O(\log{n} / \log{\gamma})$.} \label{line:paramter-begin}
    \State{Let
      $\varepsilon' = O(\varepsilon/L)$,
      $\tau'        = O(\tau/L)$,
      $\eta         = O(\min(\etaV(s, \varepsilon', \gamma, \tau'), \etaII(\varepsilon)))$,
    } 
    \State{
      $\delta       = O(\min(\eta, \varepsilon'))$,
      $\tau''       = O(1/(L\qrandom(s, \gamma, \varepsilon', \tau')))$,
      $q            = \qR(\delta, \tau'')$,
    }
    \State{
      $\lambda      = \lambdaR(\delta, \tau'')$,
      $R            = O(s / \lambda)$,
      $\mu          = O(\lambda / \gamma)$,
    }
    \State{
      $\beta_0 = 1 / 2, \beta_{L+1} = 2\lambda$, $\beta_1, \cdots, \beta_L = (1-\beta_0 - \beta_{L+1}) / L$.
    } \label{line:paramter-end}
    \State{Choose $\alpha \in [1, \gamma]$ uniformly at random.} \label{line:choose-alpha}

    \For{$i=0,\dots,L$}
    \State{Let $\qloop = \qloopC(\delta, \tau')$ and $\qwhich = \qwhichC(\delta, \tau')$.} \label{line:parameter-q}
      \State{Compute $\TS{G,i} = \Size_{\qloop, \qwhich}(G, i)$ and $\TS{H, i} = \Size_{\qloop, \qwhich}(H, i)$.} \label{line:compute-size}

      \If{$|\TS{G,i} - \TS{H,i}| > 2 \delta n$} \label{line:gh-compare-begin}
        \State{\ret NO} \label{line:difference-size-is-significant}
      \EndIf

      \If{$\TS{G,i} + \TS{H,i} < \varepsilon' n$}
        \State{{\bf continue}}
      \EndIf

      \State{Let $\tilde{n} = \max(\TS{G,i}, \TS{H,i}) + \delta n$.}
      \State{Invoke the procedure $\TestEach_{(\beta_i \varepsilon), s, (\beta_i \tau), \eta, \gamma, R}(i, \GP{i}, \HP{i}, \tilde{n})$
        with providing the random vertex query by $\Random_q(\cdot, i)$.
      } \label{line:invoke-tester}
      \If{the returned value is NO}
        \State{\ret NO} \label{line:test-failed}
      \EndIf \label{line:gh-compare-end}
    \EndFor
    \State{\ret YES}
  \EndProcedure
\end{algorithmic}
\end{algorithm}

\begin{theorem}
  \label{theorem:forest-isomorphism-tester}
  The procedure $\TestIsomorphism_{\varepsilon, \tau}(G, H)$ outputs the correct value with probability $1-\tau$ with query complexity $\polylog(n)$.
\end{theorem}
\begin{proof}
  We first calculate the probability that the procedure $\TestIsomorphism$ returns the correct value.
  In Line~\ref{line:choose-alpha},
  $|V(\GP{L+1}_{s,\alpha,\gamma,\mu,R})| \le \lambda n$ and $|V(\HP{L+1}_{s,\alpha,\gamma,\mu,R})| \le \lambda n$ hold
  with probability $1-O(\tau)$ by Lemma~\ref{lemma:size-of-garbage}
  and this is assumed in the following.
  In Line~\ref{line:compute-size},
  both $|\TS{G,i} - |V(\GP{i})| | \le \delta n$ and $|\TS{H,i} - |V(\HP{i})| | \le \delta n$
  hold for all $i$ with probability $1-O(\tau)$
  by Lemma~\ref{lemma:component-size} and the union bound.
  Therefore, if $\TS{G,i} \ge 2\delta n$, then $|V(\GP{i})| \ge \delta n$ holds. (The same thing holds for $\HP{i}$ as well.)
  Thus, from Lemma \ref{lemma:provide-random-access}, we can provide the random vertex query to $\GP{i}$ and $\HP{i}$ correctly
  in Line~\ref{line:invoke-tester} for every possible $i$
  with probability $1-O(\tau)$.
  Here, the procedure $\Random$ invokes the procedure $\Which$ $O(1/(\delta \tau''))$ times.
  Again, in what follows, we assume these happen.

  Suppose that $d(G, H) = 0$.
  In this case,
  the procedure returns NO in Line~\ref{line:difference-size-is-significant} for some $i$
  with probability at most $O(\tau)$.
  Further,
  the procedure returns NO in Line~\ref{line:test-failed} for some $i$
  with probability at most $O(\tau)$.
  Therefore, the procedure returns YES with probability $1-O(\tau)$.

  Next, suppose that $d(G, H) \ge \varepsilon n$.
  From Lemma~\ref{lemma:farness-test},
  there exists $i \in \bbN_{\le L+1}$ such that $d(\GP{i}, \HP{i}) \ge \beta_i \varepsilon n$ holds.
  However, since we assumed that $|V(\GP{L+1})|, |V(\HP{L+1})| \le \lambda n$ and we set $\beta_i = 2 \lambda$,
  $d(\GP{L+1}, \HP{L+1}) \ge \beta_i \varepsilon n$ will never hold.
  Thus, we can say that there exists $i \in \bbN_{\le L}$ such that one of the following holds:
  (i) $| |V(\GP{i})| - |V(\HP{i})| | > 4\delta n$, or
  (ii) $| |V(\GP{i})| - |V(\HP{i})| | \le 4\delta n$ and $d(\GP{i}, \HP{i}) \ge \beta_i \varepsilon n$.
  If (i) holds, the procedure returns NO in Line~\ref{line:difference-size-is-significant}.
  If (ii) holds,
  $|V(\GP{i})| + |V(\HP{i})| \ge \beta_i \varepsilon n$ holds
  (otherwise there exists a sequence of modifications from $V(\GP{i})$ into $V(\HP{i})$),
  and thus,
  $\TS{G,i} + \TS{H,i} \ge \varepsilon' n$ with the appropriate choice of constant factor of the parameters.
  Thus, the procedure does not pass Line~\ref{line:difference-size-is-significant}.
  By the assumption, the condition of Lemma~\ref{lemma:test-simple}
  (i.e., $\frac{\tilde{n}}{|V(\GP{i})|}, \frac{\tilde{n}}{|V(\HP{i})|} \in [1-\eta, 1]$)
  is satisfied.
  The procedure $\TestEach$ returns NO in Line~\ref{line:test-failed} with probability $1-O(\tau)$.

  Applying the union bound for all assumptions stated so far, the procedure $\TestEach$ outputs the correct value
  with probability $1-\tau$.

  Every parameter here is a polynomial in $\varepsilon$ and $L = O(\log{n} / \varepsilon)$,
  whose exponent is up to $\poly(t(s))$.
  Since $s = \sP(\varepsilon)$ is a constant,
  the query complexity in Line~\ref{line:compute-size} is polynomial in $O(\log{n})$.
  Consider the query complexity in Line~\ref{line:invoke-tester}.
  When $i=0$, we invoke the procedure $\TestBounded$ with constant parameters.
  Thus, the query complexity is constant.
  When $1 \le i \le L$, the query complexity of the procedure $\TestSimple$ is polynomial in the parameters.
  Therefore, the query complexity of $\TestIsomorphism$ is $\polylog(n)$ in total.
\end{proof}

\section{Proof of Theorem~\ref{thm:every-property-is-testable}}
\label{appendix:every-property-is-testable}
In this section, we prove that every property is testable with query complexity $\polylog(n)$.

\begin{proof}[of Theorem~\ref{thm:every-property-is-testable}]
  Suppose that we are given an oracle access to the input graph $G$.
  Let $\calF$ be the family of graphs that satisfy a property $P$.
  Consider the following procedure:
  \begin{enumerate}
    \item Use the parameters defined in Line~\ref{line:paramter-begin}--\ref{line:paramter-end},\ref{line:parameter-q} of Algorithm~\ref{algorithm:test-isomorhism}
      and choose $\alpha \in [1, \gamma]$ uniformly at random.
    \item For each $i \in \bbN_{\le L}$, compute $\TS{G,i} = \Size_{\qloop, \qwhich}(G, i)$.
      If $\TS{G, i} \ge O(\varepsilon' n)$ for $1 \le i \le L$,
      compute $\widetilde{\Sketch}(\GP{i}_{s, \alpha, \gamma, \mu, R})$.
      Let $\tilde{S}^{[i]}_G = \ext(\widetilde{\Sketch}(\GP{i}), M)$,
      where $\ext(\cdot)$ is an extension of a weighted point set and $M$ is a sufficiently large value.
      Similarly, 
      if $\TS{G, 0} \ge O(\varepsilon' n)$,
      compute an approximation to $\Freq_{\DII}(\GP{0})$.
    \item For each $H \in \calF$ and $i \in \bbN_{\le L}$,
      compute $z_{H,i} = |V(H^{[i]})|$ and $F_H^{[i]} = \ext(\Freq(H^{[i]}), M)$.
      Note that we know the full information of $\calF$, and therefore, we do not need to make any query to $H \in \calF$.
      Then, test isomorphism between $G$ and $H$ in the similar manner as in Line~\ref{line:gh-compare-begin}--\ref{line:gh-compare-end}.
      If $|\TS{G, i} - z_{H, i}| > 2 \delta n$, then regard that $G$ and $H$ are far from isomorphic.
      Otherwise, if $|\TS{G, i} + z_{H, i}| \ge \varepsilon' n$, we test isomorphism by the sketch of $\GP{i}$ and the weighted point set of $\HP{i}$ as follows:
      If $1 \le i \le L$, 
      compute $\MM(\tilde{S}^{[i]}_G, F_H^{[i]})$.
      If it is sufficiently large, then regard that $G$ and $H$ are far from isomorphic.
      We perform the same thing if $i=0$.
    \item If there exists $H \in \calF$ such that,
      for every $i \in \bbN_{\le L}$,
      we do not regard that $G$ and $H$ are far from isomorphic, then return YES.
      Otherwise, return NO.
  \end{enumerate}
  By the almost same argument as the proof of Theorem~\ref{theorem:forest-isomorphism-tester},
  the procedure returns YES with high probability if $G \in \calF$.
  We show that the procedure returns NO with high probability if $G$ is $\varepsilon$-far from the property $P$.
  Let $F_G^{[i]} = \Freq(\GP{i})$.
  From Lemma~\ref{lemma:mmgh-dist} and Lemma~\ref{lemma:farness-test},
  for every $H \in \calF$,
  $\MM(F_G^{[i]}, F_H^{[i]}) = \Omega(\beta_i \varepsilon n)$ holds for some $ 1 \le i \le L$
  (or $\|\Freq_{\DII}(\GP{0}) - \Freq_{\DII}(\HP{0})\|$ is sufficiently large for $i=0$).
  Assume that
  $\MM(F_G^{[i]}, \tilde{S}_G^{[i]})$ is sufficiently small.
  This happens with high probability.
  Then, from the triangle inequality,
  $\MM(\tilde{S}_G^{[i]}, F_H^{[i]})$
  is at least $\Omega(\beta_i \varepsilon n)$.
  The same argument holds for $i=0$.
  Thus, the procedure will return NO with high probability.

  The query complexity of this procedure is $\polylog(n)$.
\end{proof}

\section{Lower Bounds}\label{sec:lower-bound}
In this section, we give an $\Omega(\sqrt{\log{n}})$ lower bound for testing forest-isomorphism and prove Theorem~\ref{thm:lower-bound}.

We first mention one technical issue to show lower bounds.
Since $H$-isomorphism is a property which is closed under relabeling of vertices,
we can assume that a tester for $H$-isomorphism does not exploit labels of vertices (see~\cite{Goldreich:2002bn} for details).
Instead, we assume that a tester obtains vertices by sampling vertices uniformly at random and only asks degrees and neighbors of sampled vertices.

We introduce several definitions for probability distributions.
For two distributions $\calD_1$ and $\calD_2$ over $S$, the \emph{total variation distance} between $\calD_1$ and $\calD_2$ is defined as $\dist(\calD_1,\calD_2) = \frac{1}{2}\sum_{i \in S}|\calD_1(i) - \calD_2(i)|$.
For a set of elements $S$, we define $\calU(S)$ as the uniform distribution over $S$.

We use the following lower bound as our starting point.
\begin{lemma}[Folklore]\label{lem:uniform-is-hard}
  Suppose that a probability distribution $\calD$ over $[s]$ is given as an oracle.
  That is, upon a query, we can sample an element from the distribution $\calD$.
  We need $\Omega(\sqrt{s})$ queries to distinguish the case that $\calD = \calU([s])$ from the case that $\calD = \calU(S)$ for some $S \subseteq [s]$ with $|S| = \frac{s}{2}$.
\end{lemma}

Now, we give a way of constructing a graph from a uniform distribution.
To this end, we introduce a gadget.
For an integer $k \geq 2$, let $T_k$ be the \emph{star graph} of $k$ vertices.
That is, the vertex set of $T_k$ consists of a vertex $v$, called the \emph{center vertex}, and vertices $u_1,\ldots,u_{k-1}$ connecting to $v$.
For two integers $N$ and $k$ such that $N$ is a multiple of $k$,
we define $T_k^N$ as the (disconnected) graph consisting of $\frac{N}{k}$ copies of $T_k$.

In what follows, we fix an integer $N$ to be a huge power of two and $s = \log_2 N$ be an integer.
From a uniform distribution $\calU$ over $S \subseteq [s]$,
we construct its associated graph $G_{\calD}$ by adding a copy of $T_{2^i}^N$ to $G_\calD$ for each $i \in S$.
Note that the number of vertices in $G_\calD$ is $|S|N$,
and $T^N_{2^i}$ is well-defined since $2^i \leq 2^s = N$.

\begin{lemma}\label{lem:hard-to-distinguish-1}
  Suppose that a graph $G$ is given as an oracle in the adjacency list model.
  We need $\Omega(\sqrt{s})$ queries to distinguish the case that $G = G_{\calU([s])}$ from the case that $G = G_{\calU(S)}$ for some $S \subseteq [s]$ with $|S| = \frac{s}{2}$.
\end{lemma}
\begin{proof}
  Given an oracle access to a probability distribution $\calD$,
  which is guaranteed to be a uniform distribution over some set,
  we construct an oracle access to the graph $G$ as follows.

  \begin{description}
  \setlength{\itemsep}{0pt}
  \item[Random-vertex query:]
    We sample an element from $\calD$ and let $i$ be the output.
    Then, we construct a graph $T^N_{2^i}$ and return a random vertex in it.
    When we sample the same element $i$ again, we reuse the same $T^N_{2^i}$.
  \item[Degree query:]
    Let $v$ be the specified vertex.
    Since $v$ is a vertex returned by a random-vertex query, we know which $T^N_{2^i}$ contains the vertex $v$ and how we have choosen $v$ in $T^N_{2^i}$.
    Thus, we can return its degree.
  \item[Neighbor query:]
    Let $v$ and $i$ be the specified vertex and index, respectively.
    From the same reason as the previous case,
    we can return the $i$-th neighbor of $v$.
  \end{description}
  Note that the graph $G$ behind the oracle we have designed is equal to $G_{\calU(S)}$ when $\calD = \calU(S)$.
  Thus, from Lemma~\ref{lem:uniform-is-hard},
  we have a lower bound of $\Omega(\sqrt{s})$ on the query complexity.
\end{proof}

To obtain a lower bound for testing forest-isomorphism,
we need to show that distinguishing two forests of the same number of vertices is hard.
To address this issue, we use the following auxiliary lemma.
For a graph $G$, we define $G^{\otimes 2}$ as the graph consisting of two copies of $G$.
\begin{lemma}\label{lem:hard-to-distinguish-2}
  Suppose that a graph $G$ is given as an oracle in the adjacency list model.
  For a subset $S \subseteq [s]$ with $|S| = \frac{s}{2}$,
  we need $\Omega(\sqrt{s})$ queries to distinguish the case that $G = G_{\calU(S)}$ from the case that $G = G^{\otimes 2}_{\calU(S)}$.
\end{lemma}
\begin{proof}
  The query-answer history of an algorithm is the subgraph obtained through the interaction to the oracle.
  As long as (the distribution of) the query-answer history is the same,
  (the distribution of) the output by the algorithm is the same (See, e.g.,~\cite{Goldreich:2002bn}).
  We can assume that the query-answer history does not have labels on vertices since we are assuming that algorithms do not depend on labels of vertices.

  For each $i \in S$,
  $G^{\otimes 2}_{\calU(S)}$ contains two copies of $T^N_{2^i}$.
  It is easy to see that the distribution of the query-answer history is the same as long as an algorithm does not hit vertices from both copies of $T^N_{2^i}$.
  Suppose that we have obtained a vertex from a copy of $T^N_{2^i}$ for some $i \in S$.
  The only way to obtain a vertex from the other copy of $T^N_{2^i}$ is querying random vertices.
  Thus from the birthday paradox, we need $\Omega(\sqrt{s})$ queries to obtain vertices from both copies of $T^N_{2^i}$ for some $i \in S$.
\end{proof}

Since the number of vertices in $G_{\calU([s])}$ and $G^{\otimes 2}_{\calU(S)}$ is $sN = s2^s$,
the value $s$ is bounded from below by $\log{n} - \log{s} = \Omega(\log{n})$,
where $n = sN$.
Thus, we have the following.
\begin{corollary}\label{cor:hard-to-distinguish-3}
  Suppose that a graph $G$ is given as an oracle in the adjacency list model.
  We need $\Omega(\sqrt{\log{n}})$ queries to distinguish the case that $G = G_{\calU([s])}$ from the case that $G = G^{\otimes 2}_{\calU(S)}$ for some $S \subseteq [s]$ with $|S| = \frac{s}{2}$.
\end{corollary}

Now we show a lower bound for forest isomorphism.
From Corollary~\ref{cor:hard-to-distinguish-3},
we know that we need $\Omega(\sqrt{\log{n}})$ queries to distinguish the case $G = G_{\calU([s])}$ from the case that $G = G^{\otimes 2}_{\calU(S)}$ for some $S \subseteq [s]$ with $|S| = \frac{s}{2}$.
In the former case, $G$ is isomorphic to $H$.
We finish the proof of Theorem~\ref{thm:lower-bound} by showing that $G$ and $H$ are indeed far in the latter case.

The following lemma is useful to bound the distance between two graphs.
\begin{lemma}\label{lem:distance-by-degree}
  Let $G = (V_1,E_1)$ and $H = (V_2,E_2)$ be two graphs of $n$ vertices.
  Then,
  \begin{align*}
    \dist(G,H) \geq \min_{\phi:V_1 \to V_2}\frac{1}{2}\sum_{u \in V_1}|\deg(u) - \deg(\phi(u))|,
  \end{align*}
  where $\phi$ is over a bijection from $V_1$ to $V_2$.
\end{lemma}
\begin{proof}
  Let $\phi^*$ be a minimizer. 
  For a vertex $u \in V_1$, we define $F(u)$ as the set of edges $(u,v) \in E_1$ incident to $u$ such that $(\phi^*(u),\phi^*(v))$ is not an edge of $E_2$.
  Clearly, $|F(u)| \geq |\deg(u) - \deg(\phi(u))|$ holds for every $u$.
  The lemma holds as $\dist(G,H) = \frac{1}{2}\sum_{u \in V_1}|F(u)|$.
\end{proof}

\begin{lemma}[Lemma~3 of \cite{Wu:2013yw}]\label{lem:map-component}
  Let $G_1$ and $G_2$ be graphs.
  If some connected component $C_1$ in $G_1$ is isomorphic to a connected component $C_2$ in $G_2$,
  then we can assume that $C_1$ is mapped to $C_2$ in an optimal bijection between $G_1$ and $G_2$.
\end{lemma}

\begin{lemma}\label{lem:distance}
  Let $S$ be a subset of $[s]$ with $|S| = \frac{s}{2}$.
  Then $\dist(G_{\calU([s])}, G^{\otimes 2}_{\calU(S)}) \geq \frac{n}{8}$.
\end{lemma}
\begin{proof}
  For notational simplicity,
  let $G = G_{\calU([s])}$ and $H = G^{\otimes 2}_{\calU(S)}$.
  From Lemma~\ref{lem:map-component},
  in the optimal bijection from $G$ and $H$,
  we can assume that for each $i \in S$,
  $T^n_{2^i}$ in $G_{\calU([s])}$ is mapped to the first copy of $T^n_{2^i}$ in $G^{\otimes 2}_{\calU(S)}$.
  Let $G'$ and $H'$ be the graph obtained from $G$ and $H$ by removing these mapped vertices, respectively.

  Now we consider the distance from $G'$ to $H'$.
  We consider the loss caused by center vertices in stars of $G'$.
  Let $u$ be a center vertex of a star in $T^n_{2^i}$ for some $i$.
  Then, $u$ should be mapped to a vertex $v$ in $H'$ such that $\deg_{G'}(u) \leq \frac{1}{2}\deg_{H'}(v)$ or $\deg_{G'}(u) \geq 2\deg_{H'}(v)$.
  From Lemma~\ref{lem:distance-by-degree}, we have
  \begin{align*}
    \dist(G, H) 
    \geq
    \frac{1}{2}\sum_{u:\text{center vertex in } G'}\frac{1}{2}\deg_{G'}(u)
    \geq
    \frac{n}{8}.
  \end{align*}
  We have used the fact that the sum of degrees of center vertices of $G'$ is $\frac{n}{2}$.
\end{proof}

From the previous argument and Lemma~\ref{lem:distance}, we establish Theorem~\ref{thm:lower-bound}.

\end{document}